\documentclass[12pt,english,draftcls, one column]{IEEEtran}
\usepackage[T1]{fontenc}
\usepackage[latin9]{inputenc}
\usepackage{float}
\usepackage{amsthm}
\usepackage{amsmath}
\usepackage{amssymb}
\usepackage{graphicx}

\makeatletter

\floatstyle{ruled}
\newfloat{algorithm}{tbp}{loa}
\providecommand{\algorithmname}{Algorithm}
\floatname{algorithm}{\protect\algorithmname}

\theoremstyle{plain}

\theoremstyle{plain}

\ifCLASSINFOpdf
\else
\fi
\usepackage{babel}

\makeatother

\usepackage{babel}
\providecommand{\propositionname}{Proposition}
\providecommand{\theoremname}{Theorem}

\begin{document}

\title{Joint Precoding and Multivariate Backhaul Compression for the Downlink
of\\ Cloud Radio Access Networks}

\author{Seok-Hwan Park, Osvaldo Simeone, Onur Sahin and Shlomo Shamai (Shitz)
\thanks{S.-H. Park and O. Simeone are with the Center for Wireless Communications
and Signal Processing Research (CWCSPR), ECE Department, New Jersey
Institute of Technology (NJIT), Newark, NJ 07102, USA (email: \{seok-hwan.park,
osvaldo.simeone\}@njit.edu).

O. Sahin is with InterDigital Inc., Melville, New York, 11747, USA
(email: Onur.Sahin@interdigital.com).

S. Shamai (Shitz) is with the Department of Electrical Engineering,
Technion, Haifa, 32000, Israel (email: sshlomo@ee.technion.ac.il).%
}}
\maketitle
\begin{abstract}
This work studies the joint design of precoding and backhaul compression
strategies for the downlink of cloud radio access networks. In these
systems, a central encoder is connected to multiple multi-antenna
base stations (BSs) via finite-capacity backhaul links. At the central
encoder, precoding is followed by compression in order to produce
the rate-limited bit streams delivered to each BS over the corresponding
backhaul link. In current state-of-the-art approaches, the signals
intended for different BSs are compressed independently. In contrast,
this work proposes to leverage joint compression, also referred to
as multivariate compression, of the signals of different BSs in order
to better control the effect of the additive quantization noises at
the mobile stations (MSs). The problem of maximizing the weighted
sum-rate with respect to both the precoding matrix and the joint correlation
matrix of the quantization noises is formulated subject to power and
backhaul capacity constraints. An iterative algorithm is proposed
that achieves a stationary point of the problem. Moreover, in order
to enable the practical implementation of multivariate compression
across BSs, a novel architecture is proposed based on successive steps
of minimum mean-squared error (MMSE) estimation and per-BS compression.
Robust design with respect to imperfect channel state information
is also discussed. From numerical results, it is confirmed that the
proposed joint precoding and compression strategy outperforms conventional
approaches based on the separate design of precoding and compression
or independent compression across the BSs.\end{abstract}
\theoremstyle{theorem}
\newtheorem{theorem}{Theorem}
\theoremstyle{proposition}
\newtheorem{proposition}{Proposition}
\theoremstyle{lemma}
\newtheorem{lemma}{Lemma}
\theoremstyle{corollary}
\newtheorem{corollary}{Corollary}
\theoremstyle{definition}
\newtheorem{definition}{Definition}
\theoremstyle{remark}
\newtheorem{remark}{Remark}

\section{Introduction}

Cellular systems are evolving into heterogeneous networks consisting
of distributed base stations (BSs) covering overlapping areas of different
sizes, and thus the problems of interference management and cell association
are becoming complicated and challenging \cite{Andrews}. One of the
most promising solutions to these problems is given by so called \textit{cloud
radio access networks}, in which the encoding/decoding functionalities
of the BSs are migrated to a central unit. This is done by operating
the BSs as ``soft'' relays that interface with the central unit
via backhaul links used to carry only baseband signals (and not ``hard''
data information) \cite{Intel}-\cite{Marsch et al}. Cloud radio access
networks are expected not only to effectively handle the inter-cell
interference but also to lower system cost related to the deployment
and management of the BSs. However, one of the main impairments to
the implementation of cloud radio access networks is given by the
capacity limitations of the digital backhaul links connecting the
BSs and the central unit. These limitations are especially pronounced
for pico/femto-BSs, whose connectivity is often afforded by last-mile
cables \cite{Andrews}\cite{Andrews:femto}, and for BSs using wireless
backhaul links \cite{Maric}.

In the \textit{uplink} of cloud radio access networks, each BS compresses
its received signal to the central unit via its finite-capacity backhaul
link. The central unit then performs joint decoding of all the mobile
stations (MSs) based on all received compressed signals%
\footnote{In fact, joint decompression and decoding, an approach that is now
often seen as an instance of noisy network coding \cite{Lim}, is
generally advantageous \cite{Sanderovich:09}.%
}. Recent theoretical results have shown that \textit{distributed compression}
schemes \cite{ElGamal} can provide significant advantages over the
conventional approach based on independent compression at the BSs.
This is because the signals received by different BSs are statistically
correlated \cite{Sanderovich:09}-\cite{Park}, and hence distributed
source coding enables the quality of the compressed signal received
from one BS to be improved by leveraging the signals received from
the other BSs as side information. Note that the correlation among
the signals received by the BSs is particularly pronounced for systems
with many small cells concentrated in given areas. While current implementations
\cite{Alcatel}\cite{Dresden} employ conventional independent compression
across the BSs, the advantages of distributed source coding were first
demonstrated in \cite{Sanderovich:09}, and then studied in more general
settings in \cite{Sanderovich:MIMO}-\cite{Park}. Related works based
on the idea of computing a function of the transmitted codewords at
the BSs, also known as compute-and-forward, can be found in \cite{Nazer}\cite{Hong:UL}.

In the \textit{downlink} of cloud radio access networks, the central
encoder performs joint encoding of the messages intended for the MSs.
Then, it independently compresses the produced baseband signals to
be transmitted by each BS. These baseband signals are delivered via
the backhaul links to the corresponding BSs, which simply upconvert
and transmit them through their antennas. This system was studied
in \cite{Simeone}\cite{Marsch:GC}. In particular, in \cite{Simeone},
the central encoder performs dirty-paper coding (DPC) \cite{Costa}
of all MSs' signals before compression. A similar approach was studied
in \cite{Marsch:GC} by accounting for the effect of imperfect channel
state information (CSI). Reference \cite{Hong} instead proposes strategies
based on compute-and-forward, showing advantages in the low-backhaul
capacity regime and high sensitivity of the performance to the channel
parameters. For a review of more conventional strategies in which
the backhaul links are used to convey message information, rather
than the compressed baseband signals, we refer to \cite{Evans}-\cite{Simeone:monograph}.

\subsection{Contributions\label{sub:Contributions}}

In this work, we propose a novel approach for the compression on the
backhaul links of cloud radio access networks in the downlink that
can be seen as the counterpart of the distributed source coding strategy
studied in \cite{Sanderovich:MIMO}-\cite{Park} for the uplink. Moreover,
we propose the joint design of precoding and compression. A key idea
is that of allowing the quantization noise signals corresponding to
different BSs to be correlated with each other. The motivation behind
this choice is the fact that a proper design of the correlation of
the quantization noises across the BSs can be beneficial in limiting
the effect of the resulting quantization noise seen at the MSs. In
order to create such correlation, we propose to jointly compress the
baseband signals to be delivered over the backhaul links using so
called \textit{multivariate compression} \cite[Ch. 9]{ElGamal}. We
also show that, in practice, multivariate compression can be implemented
without resorting to joint compression across all BSs, but using instead
a successive compression strategy based on a sequence of Minimum Mean
Squared Error (MMSE) estimation and per-BS compression steps.

After reviewing some preliminaries on multivariate compression in
Sec. \ref{sec:Preliminaries}, we formulate the problem of jointly
optimizing the precoding matrix and the correlation matrix of the
quantization noises with the aim of maximizing the weighted sum-rate
subject to power and the backhaul constraints resulting from multivariate
compression in Sec. \ref{sec:Problem formulation}. There, we also
introduce the proposed architecture based on successive per-BS steps.
We then provide an iterative algorithm that achieves a stationary
point of the problem in Sec. \ref{sec:Joint}. Moreover, we compare
the proposed joint design with the more conventional approaches based
on independent backhaul compression \cite{Simeone}-\cite{Hong} or
on the separate design of precoding and (multivariate) quantization
in Sec. \ref{sec:Separate}. The robust design with respect to imperfect
CSI is also discussed in detail. In Sec. \ref{sec:Numerical-Results},
extensive numerical results are provided to illustrate the advantages
offered by the proposed approach. The paper is terminated with the
conclusion in Sec. \ref{sec:Conclusions}.

\textit{Notation}: We adopt standard information-theoretic definitions
for the mutual information $I(X;Y)$ between the random variables
$X$ and $Y$, conditional mutual information $I(X;Y|Z)$ between
$X$ and $Y$ conditioned on random variable $Z$, differential entropy
$h(X)$ of $X$ and conditional differential entropy $h(X|Y)$ of
$X$ conditioned on $Y$ \cite{ElGamal}. The distribution of a random
variable $X$ is denoted by $p(x)$ and the conditional distribution
of $X$ conditioned on $Y$ is represented by $p(x|y)$. All logarithms
are in base two unless specified. The circularly symmetric complex
Gaussian distribution with mean $\mbox{\boldmath${\mu}$}$ and covariance
matrix $\mathbf{R}$ is denoted by $\mathcal{CN}(\mbox{\boldmath${\mu}$},\bold{R})$.
The set of all $M\times N$ complex matrices is denoted by $\mathbb{C}^{M\times N}$,
and $\mathbb{E}(\cdot)$ represents the expectation operator. We use
the notations $\mathbf{X}\succeq\mathbf{0}$ and $\mathbf{X}\succ\mathbf{0}$
to indicate that the matrix $\mathbf{X}$ is positive semidefinite
and positive definite, respectively. Given a sequence $X_{1},\ldots,X_{m}$,
we define a set $X_{\mathcal{S}}=\{X_{j}|j\in\mathcal{S}\}$ for a
subset $\mathcal{S}\subseteq\{1,\ldots,m\}$. The operation $(\cdot)^{\dagger}$
denotes Hermitian transpose of a matrix or vector, and notation $\mathbf{\Sigma}_{\mathbf{x}}$
is used for the correlation matrix of random vector $\mathbf{x}$,
i.e., $\mathbf{\Sigma}_{\mathbf{x}}=\mathbb{E}[\mathbf{x}\mathbf{x}^{\dagger}]$;
$\mathbf{\Sigma}_{\mathbf{x},\mathbf{y}}$ represents the cross-correlation
matrix $\mathbf{\Sigma}_{\mathbf{x},\mathbf{y}}=\mathbb{E}[\mathbf{x}\mathbf{y}^{\dagger}]$;
$\mathbf{\Sigma}_{\mathbf{x}|\mathbf{y}}$ is used for the conditional
correlation matrix, i.e., $\mathbf{\Sigma}_{\mathbf{x}|\mathbf{y}}=\mathbb{E}[\mathbf{x}\mathbf{x}^{\dagger}|\mathbf{y}]$.

\begin{figure}
\centering\includegraphics[width=15cm,height=11.5cm]{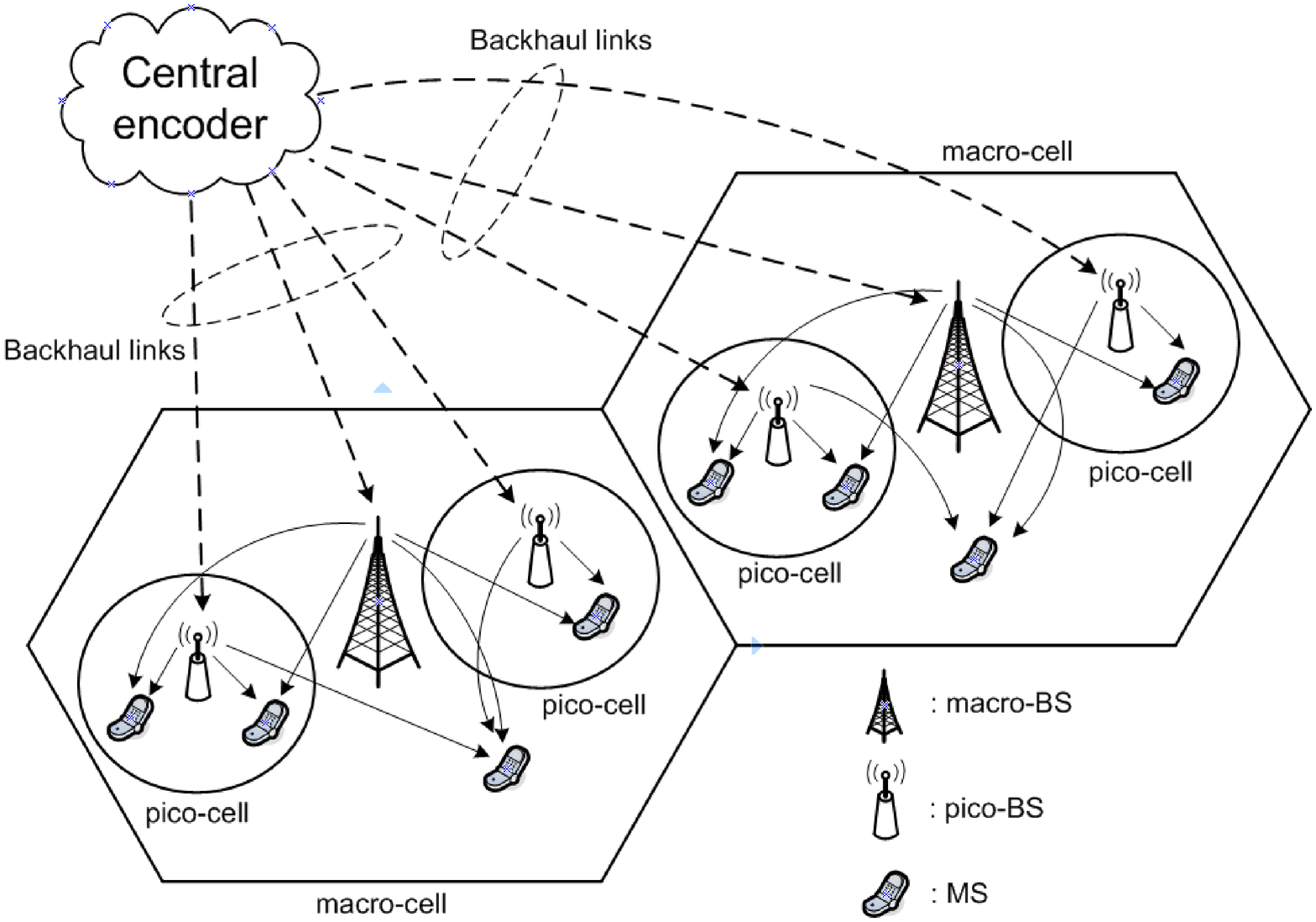}

\caption{\label{fig:block diagram}Downlink communication in a cloud radio
access network in which there are $N_{B}$ multi-antenna BSs and $N_{M}$
multi-antenna MSs. The $N_{B}$ BSs include both macro-BSs and pico/femto-BSs.
The $N_{M}$ MSs are distributed across all the cells.}
\end{figure}

\section{System Model\label{sec:System-Model}}

We consider the downlink of a cloud radio access network as illustrated
in Fig. \ref{fig:block diagram}. In the system, a central encoder
communicates to $N_{M}$ MSs through $N_{B}$ distributed BSs. The
message $M_{k}$ for each $k$th MS is uniformly distributed in the
set $\{1,\ldots,2^{nR_{k}}\}$, where $n$ is blocklength and $R_{k}$
is the information rate of message $M_{k}$ in bits per channel use
(c.u.). Each MS $k$ has $n_{M,k}$ receive antennas for $k=1,\ldots,N_{M}$,
and each BS $i$ is equipped with $n_{B,i}$ antennas for $i=1,\ldots,N_{B}$.
Note that the BSs can be either macro-BSs or pico/femto-BSs and that
the MSs are arbitrarily distributed across the cells. Each $i$th
BS is connected to the central encoder via digital backhaul link with
finite-capacity $C_{i}$ bits per c.u. For notational convenience,
we define $n_{B}=\sum_{i=1}^{N_{B}}n_{B,i}$ as the total number of
transmitting antennas, $n_{M}=\sum_{k=1}^{N_{M}}n_{M,k}$ as the total
number of receive antennas, and the sets $\mathcal{N_{B}}=\{1,\ldots,N_{B}\}$
and $\mathcal{N_{M}}=\{1,\ldots,N_{M}\}$.

As shown in Fig. \ref{fig:central encoder}, each message $M_{k}$
is first encoded by a separate channel encoder, which produces a coded
signal $\mathbf{s}_{k}$. The signal $\mathbf{s}_{k}\in\mathbb{C}^{r_{k}\times1}$
corresponds to the $r_{k}\times1$ vector of encoded symbols intended
for the $k$th MS for a given c.u., and we have $r_{k}\leq n_{M,k}$.
We assume that each coded symbol $\mathbf{s}_{k}$ is taken from a
conventional Gaussian codebook so that we have $\mathbf{s}_{k}\sim\mathcal{CN}(\mathbf{0},\mathbf{I})$.
The signals $\mathbf{s}_{1},\ldots,\mathbf{s}_{N_{M}}$ are further
processed by the central encoder in two stages, namely \textit{precoding}
and \textit{compression}. As is standard practice, precoding is used
in order to control the interference between the data streams intended
for the same MS and for different MSs. Instead, compression is needed
in order to produce the $N_{B}$ rate-limited bit streams delivered
to each BS over the corresponding backhaul link. Specifically, recall
that each BS $i$ receives up to $C_{i}$ bits per c.u. on the backhaul
link from the central encoder. Further discussion on precoding and
compression can be found in Sec. \ref{sec:Problem formulation}.

On the basis of the bits received on the backhaul links, each BS $i$
produces a vector $\mathbf{x}_{i}\in\mathbb{C}^{n_{B,i}\times1}$
for each c.u., which is the baseband signal to be transmitted from
its $n_{B,i}$ antennas. We have the per-BS power constraints%
\footnote{The results in this paper can be immediately extended to the case
with more general power constraints of the form $\mathbb{E}[\mathbf{x}^{\dagger}\mathbf{\mathbf{\Theta}}_{l}\mathbf{x}]\leq\delta_{l}$
for $l\in\{1,\ldots,L\}$, where the matrix $\mathbf{\mathbf{\Theta}}_{l}$
is a non-negative definite matrix (see, e.g., \cite[Sec. II-C]{LZhang}).%
}
\begin{equation}
\mathbb{E}\left[||\mathbf{x}_{i}||^{2}\right]\leq P_{i},\,\,\mathrm{for}\,\, i\in\mathcal{N_{B}}.\label{eq:PerBS power constraint}
\end{equation}
Assuming flat-fading channels, the signal $\mathbf{y}_{k}\in\mathbb{C}^{n_{M,k}}$
received by MS $k$ is written as
\begin{equation}
\mathbf{y}_{k}=\mathbf{H}_{k}\mathbf{x}+\mathbf{z}_{k},\label{eq:received signal MS}
\end{equation}
where we have defined the aggregate transmit signal vector $\mathbf{x}=[\mathbf{x}_{1}^{\dagger},\ldots,\mathbf{x}_{N_{B}}^{\dagger}]^{\dagger}$,
the additive noise $\mathbf{z}_{k}\sim\mathcal{CN}(\mathbf{0},\mathbf{I})$%
\footnote{Correlated noise can be easily accommodated by performing whitening
at each MS $k$ to obtain (\ref{eq:received signal MS}).%
}, and the channel matrix $\mathbf{H}_{k}\in\mathbb{C}^{n_{M,k}\times n_{B}}$
toward MS $k$ as
\begin{equation}
\mathbf{H}_{k}=\left[\mathbf{H}_{k,1}\,\,\mathbf{H}_{k,2}\,\,\cdots\,\,\mathbf{H}_{k,N_{B}}\right],\label{eq:channel to MS k}
\end{equation}
with $\mathbf{H}_{k,i}\in\mathbb{C}^{n_{M,k}\times n_{B,i}}$ denoting
the channel matrix from BS $i$ to MS $k$. The channel matrices remain
constant for the entire coding block duration. We assume that the
central encoder has information about the global channel matrices
$\mathbf{H}_{k}$ for all $k\in\mathcal{N_{M}}$ and that each MS
$k$ is only aware of the channel matrix $\mathbf{H}_{k}$. The BSs
must also be informed about the compression codebooks used by the
central encoder, as further detailed later. The case of imperfect
CSI at the central encoder will be discussed in Sec. \ref{sub:Robust-Design-with}.

Based on the definition given above and assuming single-user detection
at each MS, the rates
\begin{equation}
R_{k}=I\left(\mathbf{s}_{k};\mathbf{y}_{k}\right)\label{eq:rate MS k}
\end{equation}
can be achieved for each MS $k\in\mathcal{N_{M}}$.

\section{Preliminaries\label{sec:Preliminaries}}

This section reviews some basic information-theoretical results concerning
multivariate compression, which will be leveraged in the analysis
of the proposed backhaul compression strategy in Sec. \ref{sub:Multivariate-Backhaul-Compression}.

\begin{figure}
\centering\includegraphics[width=12.1cm,height=7.65cm]{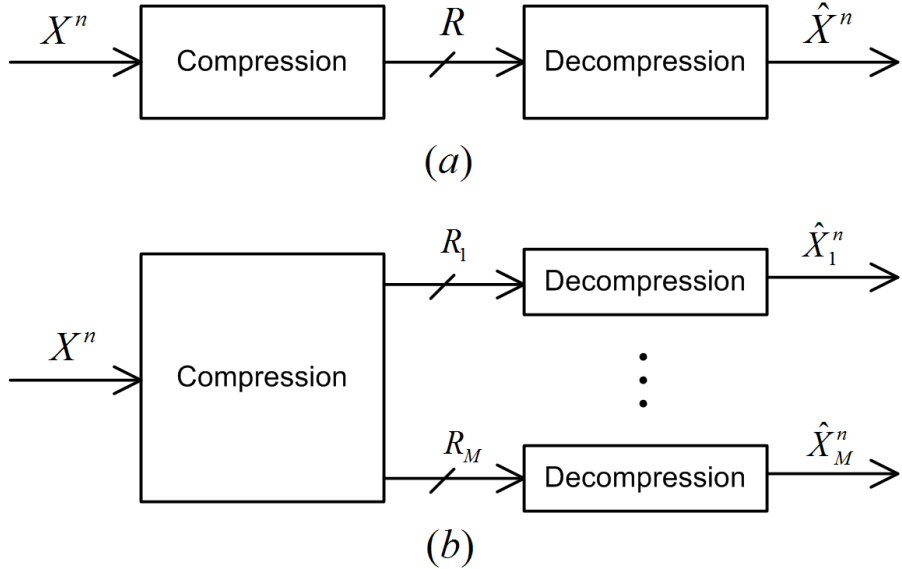}

\caption{\label{fig:compression examples}Illustration of (\textit{a}) conventional
compression; (\textit{b}) multivariate compression.}
\end{figure}

\subsection{Conventional Compression Problem\label{sub:Compression}}

To fix the ideas and the notation, we first consider the conventional
compression problem illustrated in Fig. \ref{fig:compression examples}-(a).
The compression unit compresses a random sequence $X^{n}$ of $n$
independent and identically distributed (i.i.d.) samples with distribution
$p(x)$ at a rate of $R$ bits per symbol. Specifically, the compressor
selects a codeword $\hat{X}^{n}$ within a codebook $\mathcal{C}$
of size $2^{nR}$ and sends the corresponding index, of $nR$ bits,
to the decompression unit. At the decompression unit, the sequence
$\hat{X}^{n}\in\mathcal{C}$ indicated by the received index is recovered.
Using the standard information-theoretic formulation, the compression
strategy is specified by a conditional distribution $p(\hat{x}|x)$,
which is referred to as the \textit{test channel} (see, e.g., \cite[Ch. 2]{ElGamal}).
For a given test channel, compression consists in selecting a sequence
$\hat{X}^{n}\in\mathcal{C}$ that is jointly typical%
\footnote{Two sequences $X^{n}$ and $Y^{n}$ are called jointly typical with
respect to a distribution $p(x,y)$ if their joint empirical distribution
(i.e., normalized histogram with step size $\Delta\rightarrow0$)
does not deviate much from $p(x,y)$ (see, e.g., \cite[Ch. 2]{ElGamal}
for a formal definition).%
} with the sequence $X^{n}$ with respect to the given joint distribution
$p(x,\hat{x})=p(x)p(\hat{x}|x)$. Compression is hence successful
if the encoder is able to find a jointly typical sequence $\hat{X}^{n}$
in the codebook $\mathcal{C}$. A classical result in information
theory is that this happens with arbitrarily large probability as
the block length $n$ grows large if the inequality
\begin{equation}
I\left(X;\hat{X}\right)\leq R\label{eq:conventional compression rate}
\end{equation}
is satisfied \cite[Ch. 3]{ElGamal}\cite{XZhang}.

\subsection{Multivariate Compression Problem\label{sub:Multivariate-Compression Theory}}

We now review the more general multivariate compression illustrated
in Fig. \ref{fig:compression examples}-(b). Here, the sequence $X^{n}$
is compressed into $M$ indices with the goal of producing correlated
compressed versions $\hat{X}_{1}^{n},\ldots,\hat{X}_{M}^{n}$. Each
$i$th index indicates a codeword within a codebook $\mathcal{C}_{i}$
of size $2^{nR_{i}}$, and is sent to the $i$th decompression unit
for $i\in\{1,\ldots,M\}$. Each $i$th decompression unit then recovers
a sequence $\hat{X}_{i}^{n}\in\mathcal{C}_{i}$ corresponding to the
received index. We emphasize that the choice of the codewords $\hat{X}_{1}^{n},\ldots,\hat{X}_{M}^{n}$
is done jointly at the compression unit. In particular, the specification
of the compression strategy is given by a test channel $p(\hat{x}_{1},\ldots,\hat{x}_{M}|x)$.
This implies that the compression unit wishes to find codewords $\hat{X}_{1}^{n},\ldots,\hat{X}_{M}^{n}$
that are jointly typical with the sequence $X^{n}$ with respect to
the given joint distribution $p(x,\hat{x}_{1},\ldots,\hat{x}_{M})=p(x)p(\hat{x}_{1},\ldots,\hat{x}_{M}|x)$.
The following lemma provides a sufficient condition for multivariate
compression to be successful (we refer to \cite[Lemma 8.2]{ElGamal}
for a more precise statement).

\begin{lemma}\label{lem:multivariate covering}

\emph{\cite[Ch. 9]{ElGamal}} Consider an i.i.d. sequence $X^{n}$
and $n$ large enough. Then, there exist codebooks $\mathcal{C}_{1},\ldots,\mathcal{C}_{M}$
with rates $R_{1},\ldots,R_{M}$, that have at least one tuple of
codewords $(\hat{X}_{1}^{n},\ldots,\hat{X}_{M}^{n})\in\mathcal{C}_{1}\times\ldots\times\mathcal{C}_{M}$
jointly typical with $X^{n}$ with respect to the given joint distribution
$p(x,\hat{x}_{1},\ldots,\hat{x}_{M})=p(x)p(\hat{x}_{1},\ldots,\hat{x}_{M}|x)$
with probability arbitrarily close to one, if the inequalities
\begin{equation}
\sum_{i\in\mathcal{S}}h\left(\hat{X}_{i}\right)-h\left(\hat{X}_{\mathcal{S}}|X\right)\leq\sum_{i\in\mathcal{S}}R_{i},\,\,\mathrm{for\,\, all}\,\,\mathcal{S}\subseteq\left\{ 1,\ldots,M\right\} \label{eq:multivariate covering condition}
\end{equation}
are satisfied.

\end{lemma}

\begin{proof}See \cite[Ch. 9]{ElGamal} for a proof. \end{proof}

We observe that, for a given test channel $p(\hat{x}_{1},\ldots,\hat{x}_{M}|x)$,
the inequalities (\ref{eq:multivariate covering condition}) impose
joint conditions on the rate of all codebooks. This is due to the
requirement of finding codewords $\hat{X}_{1},\ldots,\hat{X}_{M}$
that are jointly correlated according to the given test channel $p(\hat{x}_{1},\ldots,\hat{x}_{M}|x)$.
Also, note that the vector $X^{n}$ may be such that each $X_{i}$
is itself a vector and that the distortion requirements at each decompression
unit prescribe that the decompression unit be interested in only a
subset of entries in this vector. The connection between the multivariate
set-up in Fig. \ref{fig:compression examples}-(b) and the system
model under study in Fig. \ref{fig:block diagram} will be detailed
in the next section.

\section{Proposed Approach and Problem Definition\label{sec:Problem formulation}}

In this section, we first propose a novel precoding-compression strategy
based on multivariate compression for the downlink of a cloud radio
access network. We then establish the problem definition. Finally,
a novel architecture that implements multivariate compression via
a sequence of MMSE estimation and per-BS compression steps is proposed.

\subsection{Encoding Operation at the Central Encoder\label{sub:Operation-cloud encoder}}

\begin{figure}
\centering\includegraphics[width=17cm,height=7.5cm]{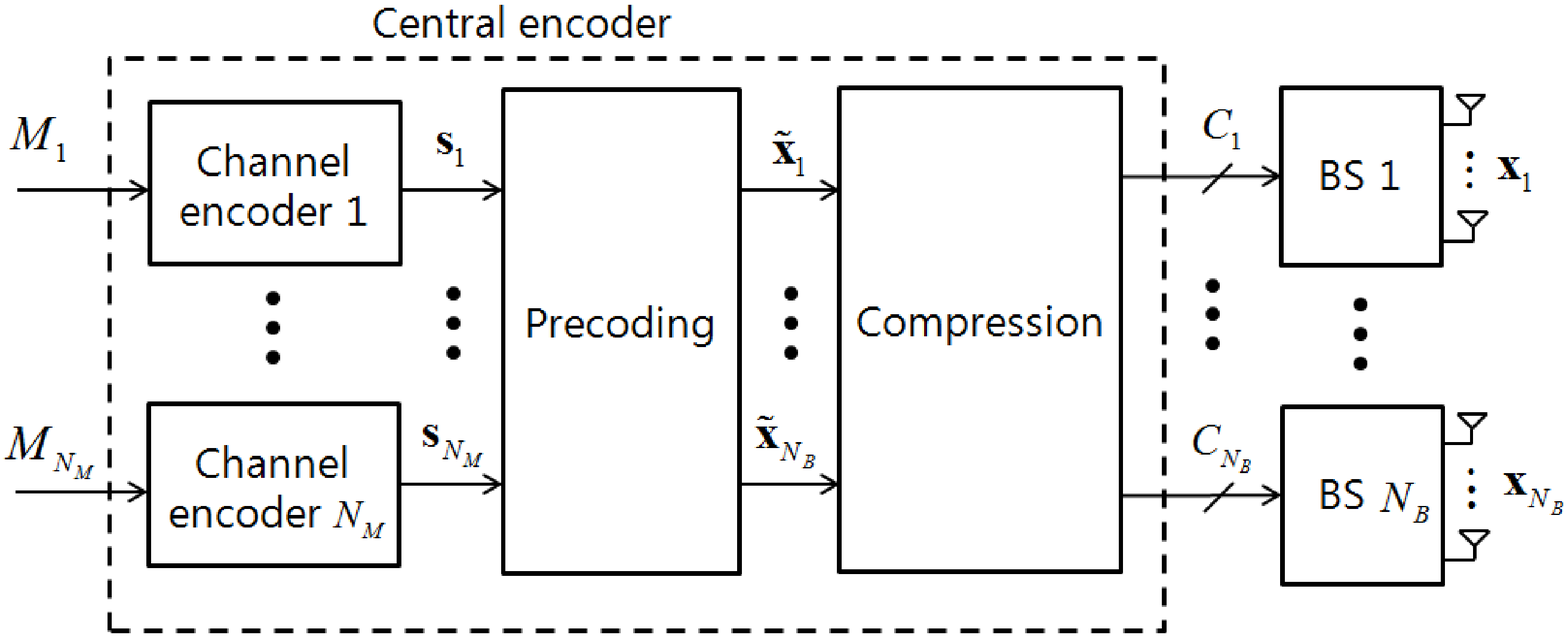}

\caption{\label{fig:central encoder}Illustration of the operation at the central
encoder.}
\end{figure}

As mentioned in the previous section, the operation at the central
encoder can be represented by the block diagram in Fig. \ref{fig:central encoder}.
Specifically, after channel encoding, the encoded signals $\mathbf{s}=[\mathbf{s}_{1}^{\dagger},\ldots,\mathbf{s}_{N_{M}}^{\dagger}]^{\dagger}$
undergo precoding and compression, as detailed next.

\textbf{1. Precoding:} In order to allow for interference management
both across the MSs and among the data streams for the same MS, the
signals in vector $\mathbf{s}$ are linearly precoded via multiplication
of a complex matrix $\mathbf{A}\in\mathbb{C}^{n_{B}\times n_{M}}$.
The precoded data can be written as
\begin{equation}
\tilde{\mathbf{x}}=\mathbf{A}\mathbf{s},\label{eq:precoding only}
\end{equation}
where the matrix $\mathbf{A}$ can be factorized as
\begin{equation}
\mathbf{A}=\left[\mathbf{A}_{1}\,\cdots\,\mathbf{A}_{N_{M}}\right],\label{eq:whole beamformer}
\end{equation}
with $\mathbf{A}_{k}\in\mathbb{C}^{n_{B}\times n_{M,k}}$ denoting
the precoding matrix corresponding to MS $k$. The precoded data $\tilde{\mathbf{x}}$
can be written as $\tilde{\mathbf{x}}=[\tilde{\mathbf{x}}_{1}^{\dagger},\ldots,\tilde{\mathbf{x}}_{N_{B}}^{\dagger}]^{\dagger}$,
where the signal $\tilde{\mathbf{x}}_{i}$ is the $n_{B,i}\times1$
precoded vector corresponding to the $i$th BS and given as
\begin{equation}
\tilde{\mathbf{x}}_{i}=\mathbf{E}_{i}^{\dagger}\mathbf{A}\mathbf{s},\label{eq:precoding BS-wise}
\end{equation}
with the matrix $\mathbf{E}_{i}\in\mathbb{C}^{n_{B}\times n_{B,i}}$
having all zero elements except for the rows from $(\sum_{j=1}^{i-1}n_{B,j}+1)$
to $(\sum_{j=1}^{i}n_{B,j})$ which contain an $n_{B,i}\times n_{B,i}$
identity matrix. Note that non-linear precoding using DPC techniques
can also be considered, as discussed in Remark \ref{rem:DPC} below.

\textbf{2. Compression:} Each precoded data stream $\tilde{\mathbf{x}}_{i}$
for $i\in\mathcal{N_{B}}$ must be compressed in order to allow the
central encoder to deliver it to the $i$th BS through the backhaul
link of capacity $C_{i}$ bits per c.u. Each $i$th BS then simply
forwards the compressed signal $\mathbf{x}_{i}$ obtained from the
central encoder. Note that this implies that the BSs need not be aware
of the channel codebooks and of the precoding matrix $\mathbf{A}$
used by the central encoder. Instead, they must be informed about
the quantization codebooks selected by the central encoder.

Using standard rate-distortion considerations, we adopt a Gaussian
test channel to model the effect of compression on the backhaul link.
In particular, we write the compressed signals $\mathbf{x}_{i}$ to
be transmitted by BS $i$ as%
\footnote{The test channel $\mathbf{x}_{i}=\mathbf{B}_{i}\tilde{\mathbf{x}}_{i}+\mathbf{q}_{i}$
is seemingly more general than (\ref{eq:Gaussian test channel each BS}),
but this can be captured by adjusting the matrix $\mathbf{A}$ in
(\ref{eq:precoding only}).%
}
\begin{align}
\mathbf{x}_{i} & =\tilde{\mathbf{x}}_{i}+\mathbf{q}_{i},\label{eq:Gaussian test channel each BS}
\end{align}
where the compression noise $\mathbf{q}_{i}$ is modeled as a complex
Gaussian vector distributed as $\mathcal{CN}(\mathbf{0},\mathbf{\Omega}_{i,i})$.
Overall, the vector $\mathbf{x}=[\mathbf{x}_{1}^{\dagger},\ldots,\mathbf{x}_{N_{B}}^{\dagger}]^{\dagger}$
of compressed signals for all the BSs is given by
\begin{equation}
\mathbf{x}=\mathbf{A}\mathbf{s}+\mathbf{q},\label{eq:whole encoding operation}
\end{equation}
where the compression noise $\mathbf{q}=[\mathbf{q}_{1}^{\dagger},\ldots,\mathbf{q}_{N_{B}}^{\dagger}]^{\dagger}$
is modeled as a complex Gaussian vector distributed as $\mathbf{q}\sim\mathcal{CN}(\mathbf{0},\mathbf{\Omega})$.
The compression covariance $\mathbf{\Omega}$ is given as
\begin{equation}
\mathbf{\Omega}=\left[\begin{array}{cccc}
\mathbf{\Omega}_{1,1} & \mathbf{\Omega}_{1,2} & \cdots & \mathbf{\Omega}_{1,N_{B}}\\
\mathbf{\Omega}_{2,1} & \mathbf{\Omega}_{2,2} & \cdots & \mathbf{\Omega}_{2,N_{B}}\\
\vdots & \vdots & \ddots & \vdots\\
\mathbf{\Omega}_{N_{B},1} & \mathbf{\Omega}_{N_{B},2} & \cdots & \mathbf{\Omega}_{N_{B},N_{B}}
\end{array}\right],\label{eq:compression covariance}
\end{equation}
where the matrix $\mathbf{\Omega}_{i,j}$ is defined as $\mathbf{\Omega}_{i,j}=\mathbb{E}[\mathbf{q}_{i}\mathbf{q}_{j}^{\dagger}]$
and defines the correlation between the quantization noises of BS
$i$ and BS $j$. Rate-distortion theory guarantees that compression
codebooks can be found for any given covariance matrix $\mathbf{\Omega}\succeq\mathbf{0}$
under appropriate constraints imposed on the backhaul links' capacities.
This aspect will be further discussed in Sec. \ref{sub:Multivariate-Compression Theory}.

With the described precoding and compression operations, the achievable
rate (\ref{eq:rate MS k}) for MS $k$ is computed as
\begin{align}
I\left(\mathbf{s}_{k};\mathbf{y}_{k}\right) & =f_{k}\left(\mathbf{A},\mathbf{\Omega}\right)\label{eq:rate MS k computed}\\
 & \triangleq\log\det\left(\mathbf{I}+\mathbf{H}_{k}\left(\mathbf{A}\mathbf{A}^{\dagger}+\mathbf{\Omega}\right)\mathbf{H}_{k}^{\dagger}\right)-\log\det\left(\mathbf{I}+\mathbf{H}_{k}\left(\sum_{l\in\mathcal{N_{M}}\setminus\{k\}}\mathbf{A}_{l}\mathbf{A}_{l}^{\dagger}+\mathbf{\Omega}\right)\mathbf{H}_{k}^{\dagger}\right).\nonumber
\end{align}

\begin{remark}\label{rem:compression}In the conventional approach
studied in \cite{Simeone}-\cite{Hong}, the signals $\tilde{\mathbf{x}}_{i}$
corresponding to each BS $i$ are compressed independently. This corresponds
to setting $\mathbf{\Omega}_{i,j}=\mathbf{0}$ for all $i\neq j$
in (\ref{eq:compression covariance}). A key contribution of this
work is the proposal to leverage correlated compression for the signals
of different BSs in order to better control the effect of the additive
quantization noises at the MSs. \end{remark}

\begin{remark}The design of the precoding matrix $\mathbf{A}$ and
of the quantization covariance $\mathbf{\Omega}$ can be either performed
separately, e.g., by using a conventional precoder $\mathbf{A}$ such
as zero-forcing (ZF) or MMSE precoding (see, e.g., \cite{RZhang}-\cite{Luo}),
or jointly. Both approaches will be investigated in the following.
\end{remark}

\begin{remark}\label{rem:DPC}If non-linear precoding via DPC \cite{Costa}
is deployed at the central encoder with a specific encoding permutation
$\pi_{\mathrm{DPC}}:\mathcal{N_{M}}\rightarrow\mathcal{N_{M}}$ of
the MS indices $\mathcal{N_{M}}$, the achievable rate $R_{\pi_{\mathrm{DPC}}(k)}$
for MS $\pi_{\mathrm{DPC}}(k)$ is given as $R_{\pi_{\mathrm{DPC}}(k)}=I\left(\mathbf{s}_{\pi_{\mathrm{DPC}}(k)};\mathbf{y}_{\pi_{\mathrm{DPC}}(k)}|\mathbf{s}_{\pi_{\mathrm{DPC}}(1)},\ldots,\mathbf{s}_{\pi_{\mathrm{DPC}}(k-1)}\right)$
in lieu of (\ref{eq:rate MS k}) and can be calculated as $R_{\pi_{\mathrm{DPC}}(k)}=\tilde{f}_{\pi_{\mathrm{DPC}}(k)}\left(\mathbf{A},\mathbf{\Omega}\right)$
with the function $\tilde{f}_{\pi_{\mathrm{DPC}}(k)}\left(\mathbf{A},\mathbf{\Omega}\right)$
given as
\begin{align}
\tilde{f}_{\pi_{\mathrm{DPC}}(k)}\left(\mathbf{A},\mathbf{\Omega}\right) & \triangleq\log\det\left(\mathbf{I}+\mathbf{H}_{\pi_{\mathrm{DPC}}(k)}\left(\sum_{l=k}^{N_{M}}\mathbf{A}_{\pi_{\mathrm{DPC}}(l)}\mathbf{A}_{\pi_{\mathrm{DPC}}(l)}^{\dagger}+\mathbf{\Omega}\right)\mathbf{H}_{\pi_{\mathrm{DPC}}(k)}^{\dagger}\right)\\
 & -\log\det\left(\mathbf{I}+\mathbf{H}_{\pi_{\mathrm{DPC}}(k)}\left(\sum_{l=k+1}^{N_{M}}\mathbf{A}_{\pi_{\mathrm{DPC}}(l)}\mathbf{A}_{\pi_{\mathrm{DPC}}(l)}^{\dagger}+\mathbf{\Omega}\right)\mathbf{H}_{\pi_{\mathrm{DPC}}(k)}^{\dagger}\right).\nonumber
\end{align}
Note that the DPC is designed based on the knowledge of the noise
levels (including the quantization noise) in order to properly select
the MMSE scaling factor \cite{Erez}.

\end{remark}

\subsection{Multivariate Backhaul Compression\label{sub:Multivariate-Backhaul-Compression}}

As explained above, due to the fact that the BSs are connected to
the central encoder via finite-capacity backhaul links, the precoded
signals $\tilde{\mathbf{x}}_{i}$ in (\ref{eq:precoding BS-wise})
for $i\in\mathcal{N_{B}}$ are compressed before they are communicated
to the BSs using the Gaussian test channels (\ref{eq:Gaussian test channel each BS}).
In the conventional case in which the compression noise signals related
to the different BSs are uncorrelated, i.e., $\mathbf{\Omega}_{i,j}=\mathbf{0}$
for all $i\neq j\in\mathcal{N_{B}}$ as in \cite{Simeone}-\cite{Hong},
the signal $\mathbf{x}_{i}$ to be emitted from BS $i$ can be reliably
communicated from the central encoder to BS $i$ if the condition
\begin{equation}
I\left(\tilde{\mathbf{x}}_{i};\mathbf{x}_{i}\right)=\log\det\left(\mathbf{E}_{i}^{\dagger}\mathbf{A}\mathbf{A}^{\dagger}\mathbf{E}_{i}+\mathbf{\Omega}_{i,i}\right)-\log\det\left(\mathbf{\Omega}_{i,i}\right)\leq C_{i}\label{eq:independent comp condition}
\end{equation}
is satisfied for $i\in\mathcal{N_{B}}$. This follows from standard
rate-distortion theoretic arguments (see, e.g., \cite{XZhang} and
Sec. \ref{sub:Compression}). We emphasize that (\ref{eq:independent comp condition})
is valid under the assumption that each BS $i$ is informed about
the quantization codebook used by the central encoder, as defined
by the covariance matrix $\mathbf{\Omega}_{i,i}$.

In this paper, we instead propose to introduce correlation among the
compression noise signals, i.e., to set $\mathbf{\Omega}_{i,j}\neq\mathbf{0}$
for $i\neq j$, in order to control the effect of the quantization
noise at the MSs. As discussed in Sec. \ref{sec:Preliminaries}, introducing
correlated quantization noises calls for joint, and not independent,
compression of the precoded signals of different BSs. As seen, the
family of compression strategies that produce descriptions with correlated
compression noises is often referred to as \textit{multivariate compression.}
By choosing the test channel according to (\ref{eq:whole encoding operation})
(see Sec. \ref{sub:Multivariate-Compression Theory}), we can leverage
Lemma \ref{lem:multivariate covering} to obtain sufficient conditions
for the signal $\mathbf{x}_{i}$ to be reliably delivered to BS $i$
for all $i\in\mathcal{N_{B}}$. In Lemma \ref{lem:multivariate},
we use $\mathbf{E}_{\mathcal{S}}$ to denote the matrix obtained by
stacking the matrices $\mathbf{E}_{i}$ for $i\in\mathcal{S}$ horizontally.

\begin{lemma}\label{lem:multivariate} The signals $\mathbf{x}_{1},\ldots,\mathbf{x}_{N_{B}}$
obtained via the test channel (\ref{eq:whole encoding operation})
can be reliably transferred to the BSs on the backhaul links if the
condition
\begin{align}
g_{\mathcal{S}}\left(\mathbf{A},\mathbf{\Omega}\right) & \triangleq\sum_{i\in\mathcal{S}}h\left(\mathbf{x}_{i}\right)-h\left(\mathbf{x}_{\mathcal{S}}|\tilde{\mathbf{x}}\right)\label{eq:multivariate computed}\\
 & =\sum_{i\in\mathcal{S}}\log\det\left(\mathbf{E}_{i}^{\dagger}\mathbf{A}\mathbf{A}^{\dagger}\mathbf{E}_{i}+\mathbf{\Omega}_{i,i}\right)-\log\det\left(\mathbf{E}_{\mathcal{S}}^{\dagger}\mathbf{\Omega}\mathbf{E}_{\mathcal{S}}\right)\leq\sum_{i\in\mathcal{S}}C_{i}\nonumber
\end{align}
is satisfied for all subsets $\mathcal{S}\subseteq\mathcal{N_{B}}$.

\end{lemma}

\begin{proof} The proof follows by applying Lemma \ref{lem:multivariate covering}
by substituting $\tilde{\mathbf{x}}=\mathbf{A}\mathbf{s}$ for the
signal $X$ to be compressed, and $\mathbf{x}_{1},\ldots,\mathbf{x}_{N_{B}}$
for the compressed versions $\hat{X}_{1},\ldots,\hat{X}_{M}$. \end{proof}

Comparing (\ref{eq:independent comp condition}) with (\ref{eq:multivariate computed})
shows that the introduction of correlation among the quantization
noises for different BSs leads to additional constraints on the backhaul
link capacities.

\subsection{Weighted Sum-Rate Maximization\label{sub:Sum-Rate-Maximization}}

Assuming the operation at the central encoder, BSs and MSs detailed
above, we are interested in maximizing the weighted sum-rate $R_{\mathrm{sum}}=\sum_{k=1}^{N_{M}}w_{k}R_{k}$
subject to the backhaul constraints (\ref{eq:multivariate computed})
over the precoding matrix $\mathbf{A}$ and the compression noise
covariance $\mathbf{\Omega}$ for given weights $w_{k}\geq0$, $k\in\mathcal{N_{M}}$.
This problem is formulated as \begin{subequations}\label{eq:original problem}
\begin{align}
\underset{\mathbf{A},\,\mathbf{\Omega}\succeq\mathbf{0}}{\mathrm{maximize}} & \,\,\sum_{k=1}^{N_{M}}w_{k}f_{k}\left(\mathbf{A},\mathbf{\Omega}\right)\label{eq:original problem objective}\\
\mathrm{s.t.}\,\,\,\,\,\,\,\,\, & g_{\mathcal{S}}\left(\mathbf{A},\mathbf{\Omega}\right)\leq\sum_{i\in\mathcal{S}}C_{i},\,\,\mathrm{for\,\, all}\,\,\mathcal{S}\subseteq\mathcal{N_{B}},\label{eq:original problem backhaul}\\
 & \mathrm{tr}\left(\mathbf{E}_{i}^{\dagger}\mathbf{A}\mathbf{A}^{\dagger}\mathbf{E}_{i}+\mathbf{\Omega}_{i,i}\right)\leq P_{i},\,\,\mathrm{for\,\, all}\,\, i\in\mathcal{N_{B}}.\label{eq:original problem power}
\end{align}
\end{subequations}The condition (\ref{eq:original problem backhaul})
corresponds to the backhaul constraints due to multivariate compression
as introduced in Lemma \ref{lem:multivariate}, and the condition
(\ref{eq:original problem power}) imposes the transmit power constraints
(\ref{eq:PerBS power constraint}). It is noted that the problem (\ref{eq:original problem})
is not easy to solve due to the non-convexity of the objective function
$\sum_{k=1}^{N_{M}}w_{k}f_{k}\left(\mathbf{A},\mathbf{\Omega}\right)$
in (\ref{eq:original problem objective}) and the functions $g_{\mathcal{S}}\left(\mathbf{A},\mathbf{\Omega}\right)$
in (\ref{eq:original problem backhaul}) with respect to $(\mathbf{A},\mathbf{\Omega})$.
In Sec. \ref{sec:Joint}, we will propose an algorithm to tackle the
solution of problem (\ref{eq:original problem}).

\subsection{Successive Estimation-Compression Architecture\label{sub:Successive-Compression}}

\begin{figure}
\centering\includegraphics[width=16.5cm,height=9.2cm]{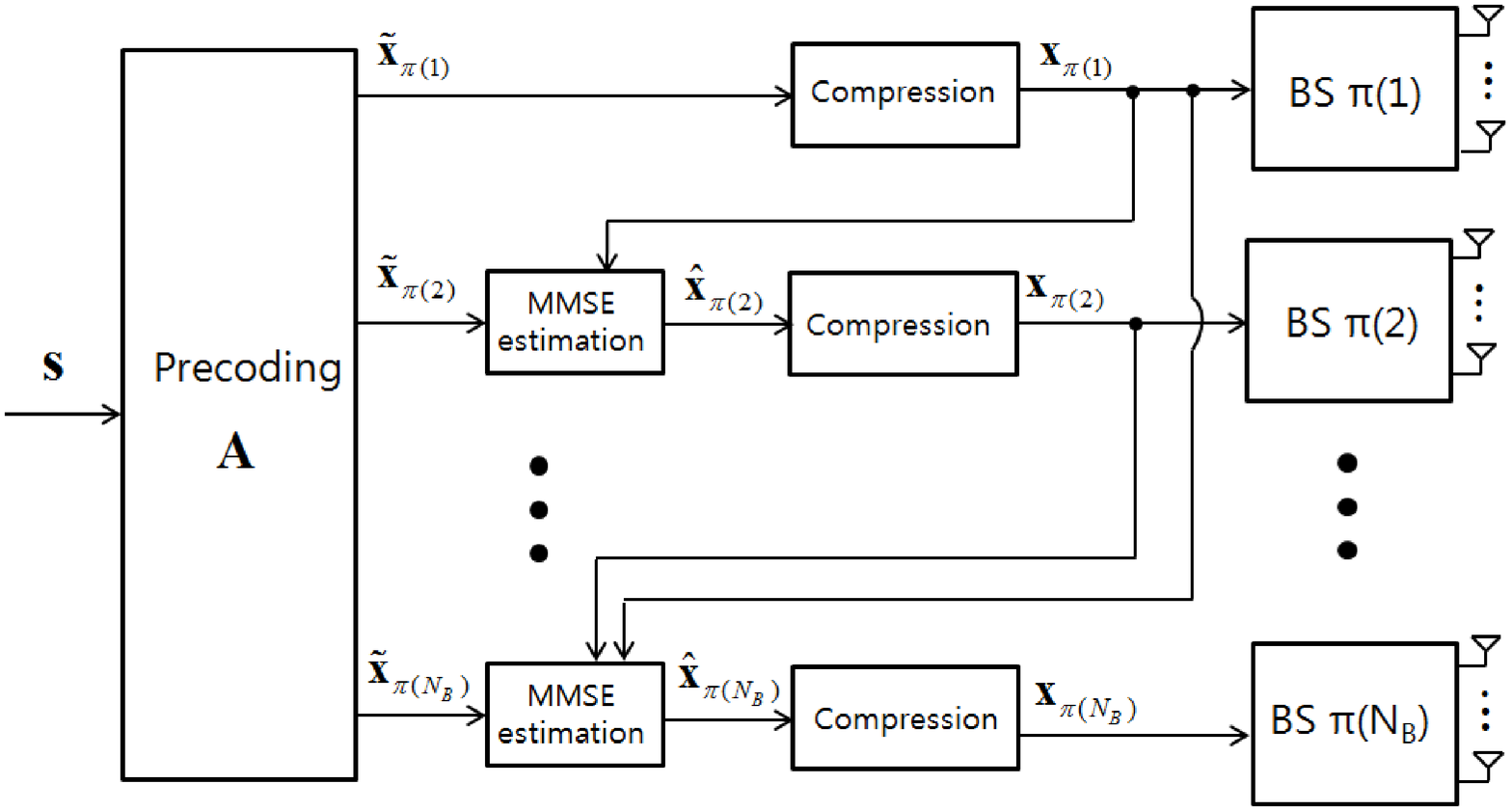}

\caption{\label{fig:successive comp}Proposed architecture for multivariate
compression based on successive steps of MMSE estimation and per-BS
compression.}
\end{figure}

In order to obtain correlated quantization noises across BSs using
multivariate compression, it is in principle necessary to perform
joint compression of all the precoded signals $\tilde{\mathbf{x}}_{i}$
corresponding to all BSs $i$ for $i\in\mathcal{N_{B}}$ (see Sec.
\ref{sub:Multivariate-Compression Theory}). If the number of BSs
is large, this may easily prove to be impractical. Here, we argue
that, in practice, joint compression is not necessary and that the
successive strategy based on MMSE estimation and per-BS compression
illustrated in Fig. \ref{fig:successive comp} is sufficient. The
proposed approach works with a fixed permutation $\pi:\mathcal{N_{B}}\rightarrow\mathcal{N_{B}}$
of the BSs' indices $\mathcal{N_{B}}$.

The central encoder first compresses the signal $\tilde{\mathbf{x}}_{\pi(1)}$
using the test channel (\ref{eq:Gaussian test channel each BS}),
namely $\mathbf{x}_{\pi(1)}=\tilde{\mathbf{x}}_{\pi(1)}+\mathbf{q}_{\pi(1)}$,
with $\mathbf{q}_{\pi(1)}\sim\mathcal{CN}(\mathbf{0},\mathbf{\Omega}_{\pi(1),\pi(1)})$,
and sends the bit stream describing the compressed signal $\mathbf{x}_{\pi(1)}$
over the backhaul link to BS $\pi(1)$. Then, for any other $i\in\mathcal{N_{B}}$
with $i>1$, the central encoder obtains the compressed signal $\mathbf{x}_{\pi(i)}$
for BS $\pi(i)$ in a successive manner in the given order $\pi$
by performing the following steps:

\textbf{(}\textbf{\textit{a}}\textbf{)} \textbf{\textit{Estimation}}\textbf{:}
The central encoder obtains the MMSE estimate $\mathbf{\hat{x}}_{\pi(i)}$
of $\mathbf{x}_{\pi(i)}$ given the signal $\tilde{\mathbf{x}}_{\pi(i)}$
and the previously obtained compressed signals $\mathbf{x}_{\pi(1)},\ldots,\mathbf{x}_{\pi(i-1)}$.
This estimate is given by
\begin{align}
\hat{\mathbf{x}}_{\pi(i)} & =\mathbb{E}\left[\mathbf{x}_{\pi(i)}|\mathbf{u}_{\pi(i)}\right]\label{eq:MMSE estimate x_hat_i}\\
 & =\mathbf{\Sigma}_{\mathbf{x}_{\pi(i)},\mathbf{u}_{\pi(i)}}\mathbf{\Sigma}_{\mathbf{u}_{\pi(i)}}^{-1}\mathbf{u}_{\pi(i)},\nonumber
\end{align}
where we defined the vector $\mathbf{u}_{\pi(i)}=[\mathbf{x}_{\pi(1)}^{\dagger},\ldots,\mathbf{x}_{\pi(i-1)}^{\dagger},\tilde{\mathbf{x}}_{\pi(i)}^{\dagger}]^{\dagger}$,
and the correlation matrices $\mathbf{\Sigma}_{\mathbf{x}_{\pi(i)},\mathbf{u}_{\pi(i)}}$
and $\mathbf{\Sigma}_{\mathbf{u}_{\pi(i)}}$ are given as
\begin{align}
\mathbf{\Sigma}_{\mathbf{x}_{\pi(i)},\mathbf{u}_{\pi(i)}} & =\left[\left(\mathbf{E}_{\pi(i)}^{\dagger}\mathbf{A}\mathbf{A}^{\dagger}\mathbf{E}_{\mathcal{S}_{\pi,i-1}}+\mathbf{\Omega}_{\pi(i),\mathcal{S}_{\pi,i-1}}\right)\,\,\mathbf{E}_{\pi(i)}^{\dagger}\mathbf{A}\mathbf{A}^{\dagger}\mathbf{E}_{\pi(i)}\right]
\end{align}
and
\begin{align}
\mathbf{\Sigma}_{\mathbf{u}_{\pi(i)}} & =\left[\begin{array}{cc}
\mathbf{E}_{\mathcal{S}_{\pi,i-1}}^{\dagger}\mathbf{A}\mathbf{A}^{\dagger}\mathbf{E}_{\mathcal{S}_{\pi,i-1}}+\mathbf{\Omega}_{\mathcal{S}_{\pi,i-1},\mathcal{S}_{\pi,i-1}} & \mathbf{E}_{\mathcal{S}_{\pi,i-1}}^{\dagger}\mathbf{A}\mathbf{A}^{\dagger}\mathbf{E}_{\pi(i)}\\
\mathbf{E}_{\pi(i)}^{\dagger}\mathbf{A}\mathbf{A}^{\dagger}\mathbf{E}_{\mathcal{S}_{\pi,i-1}} & \mathbf{E}_{\pi(i)}^{\dagger}\mathbf{A}\mathbf{A}^{\dagger}\mathbf{E}_{\pi(i)}
\end{array}\right],
\end{align}
with $\mathbf{\Omega}_{\mathcal{S},\mathcal{T}}=\mathbf{E}_{\mathcal{S}}^{\dagger}\mathbf{\Omega}\mathbf{E}_{\mathcal{T}}$
for subsets $\mathcal{S},\mathcal{T}\subseteq\mathcal{N_{B}}$ and
the set $\mathcal{S}_{\pi,i}$ defined as $\mathcal{S}_{\pi,i}\triangleq\{\pi(1),\ldots,\pi(i)\}$.

\textbf{(}\textbf{\textit{b}}\textbf{)} \textbf{\textit{Compression}}\textbf{:}
The central encoder compresses the MMSE estimate $\hat{\mathbf{x}}_{\pi(i)}$
to obtain $\mathbf{x}_{\pi(i)}$ using the test channel
\begin{equation}
\mathbf{x}_{\pi(i)}=\hat{\mathbf{x}}_{\pi(i)}+\hat{\mathbf{q}}_{\pi(i)},\label{eq:compression after MMSE}
\end{equation}
where the quantization noise $\hat{\mathbf{q}}_{\pi(i)}$ is independent
of the estimate $\hat{\mathbf{x}}_{\pi(i)}$ and distributed as $\hat{\mathbf{q}}_{\pi(i)}\sim\mathcal{CN}(\mathbf{0},\mathbf{\Sigma}_{\mathbf{x}_{\pi(i)}|\hat{\mathbf{x}}_{\pi(i)}})$
with
\begin{align}
\mathbf{\Sigma}_{\mathbf{x}_{\pi(i)}|\hat{\mathbf{x}}_{\pi(i)}} & =\mathbf{\Sigma}_{\mathbf{x}_{\pi(i)}|\mathbf{u}_{\pi(i)}}\label{eq:conditional covariance}\\
 & =\mathbf{\Omega}_{\pi(i),\pi(i)}-\mathbf{\Omega}_{\pi(i),\mathcal{S}_{\pi,i-1}}\mathbf{\Omega}_{\mathcal{S}_{\pi,i-1},\mathcal{S}_{\pi,i-1}}^{-1}\mathbf{\Omega}_{\pi(i),\mathcal{S}_{\pi,i-1}}^{\dagger}.\nonumber
\end{align}
Note that the first equality in (\ref{eq:conditional covariance})
follows from the fact that the MMSE estimate $\hat{\mathbf{x}}_{\pi(i)}$
is a sufficient statistic for the estimation of $\mathbf{x}_{\pi(i)}$
from $\mathbf{u}_{\pi(i)}$ (see, e.g., \cite{Forney}). Moreover,
the compression rate $I(\hat{\mathbf{x}}_{\pi(i)};\mathbf{x}_{\pi(i)})$
required by the test channel (\ref{eq:compression after MMSE}) is
given by
\begin{align}
I(\mathbf{x}_{\pi(i)};\hat{\mathbf{x}}_{\pi(i)})= & h\left(\mathbf{x}_{\pi(i)}\right)-h\left(\mathbf{x}_{\pi(i)}|\hat{\mathbf{x}}_{\pi(i)}\right)\label{eq:compression rate MMSE compression}\\
= & \log\det\left(\mathbf{E}_{\pi(i)}^{\dagger}\mathbf{A}\mathbf{A}^{\dagger}\mathbf{E}_{\pi(i)}+\mathbf{\Omega}_{\pi(i),\pi(i)}\right)\nonumber \\
 & -\log\det\left(\mathbf{\Omega}_{\pi(i),\pi(i)}-\mathbf{\Omega}_{\pi(i),\mathcal{S}_{\pi,i-1}}\mathbf{\Omega}_{\mathcal{S}_{\pi,i-1},\mathcal{S}_{\pi,i-1}}^{-1}\mathbf{\Omega}_{\pi(i),\mathcal{S}_{\pi,i-1}}^{\dagger}\right).\nonumber
\end{align}

To see why the structure in Fig. \ref{fig:successive comp} described
above realizes multivariate compression, we need the following lemma.

\begin{lemma}\label{lem:contra polymatroid} The region of the backhaul
capacity tuples $(C_{1},\ldots,C_{N_{B}})$ satisfying the constraints
(\ref{eq:original problem backhaul}) is a \textit{contrapolymatroid}%
\footnote{Let us define $\mathcal{M}=\{1,\ldots,M\}$ and $f:\tilde{\mathcal{M}}\rightarrow\mathbb{R}_{+}$
with $\tilde{\mathcal{M}}$ being the set of all subsets of $\mathcal{M}$.
Then, the polyhedron $\{(x_{1},\ldots,x_{M})|\sum_{i\in\mathcal{S}}x_{i}\geq f(\mathcal{S}),\,\,\mathrm{for\,\, all}\,\,\mathcal{S}\subseteq\mathcal{M}\}$
is a contrapolymatroid if the function $f$ satisfies the conditions:
(\textit{a}) $f(\emptyset)=0$; (\textit{b}) $f(\mathcal{S})\leq f(\mathcal{T})$
if $\mathcal{S}\subset\mathcal{T}$; (\textit{c}) $f(\mathcal{S})+f(\mathcal{T})\leq f(\mathcal{S}\cup\mathcal{T})+f(\mathcal{S}\cap\mathcal{T})$
\cite[Def. 3.1]{Tse}.%
} \cite[Def. 3.1]{Tse}. Therefore, it has a corner point for each
permutation $\pi$ of the BS indices $\mathcal{N_{B}}$, and each
such corner point is given by the tuple $(C_{\pi(1)},\ldots,C_{\pi(N_{B})})$
with
\begin{align}
C_{\pi(i)}= & I\left(\mathbf{x}_{\pi(i)};\tilde{\mathbf{x}},\mathbf{x}_{\pi(1)},\ldots,\mathbf{x}_{\pi(i-1)}\right)\label{eq:corner points}\\
= & I(\mathbf{x}_{\pi(i)};\hat{\mathbf{x}}_{\pi(i)})\nonumber
\end{align}
for $i=1,\ldots,N_{B}$. Moreover, the corner point $(C_{\pi(1)},\ldots,C_{\pi(N_{B})})$
in (\ref{eq:corner points}) is such that the constraints (\ref{eq:original problem backhaul})
are satisfied with equality for the subsets $\mathcal{S}=\{\pi(1)\},\{\pi(1),\pi(2)\},\ldots,$
$\{\pi(1),\ldots,\pi(N_{B})\}$.

\end{lemma}

\begin{proof} This lemma follows immediately by the definition and
properties of contrapolymatroids as summarized in \cite[Def. 3.1]{Tse}.
Moreover, the second equality of (\ref{eq:corner points}) holds due
to the fact that the MMSE estimate $\hat{\mathbf{x}}_{\pi(i)}$ is
a sufficient statistic for the estimation of $\mathbf{x}_{\pi(i)}$
from $\mathbf{u}_{\pi(i)}$ (see, e.g., \cite{Forney}), or equivalently
from the Markov chain $\mathbf{x}_{\pi(i)}-\hat{\mathbf{x}}_{\pi(i)}-\mathbf{u}_{\pi(i)}$.\end{proof}

\begin{figure}
\centering\includegraphics[width=12cm,height=9.2cm]{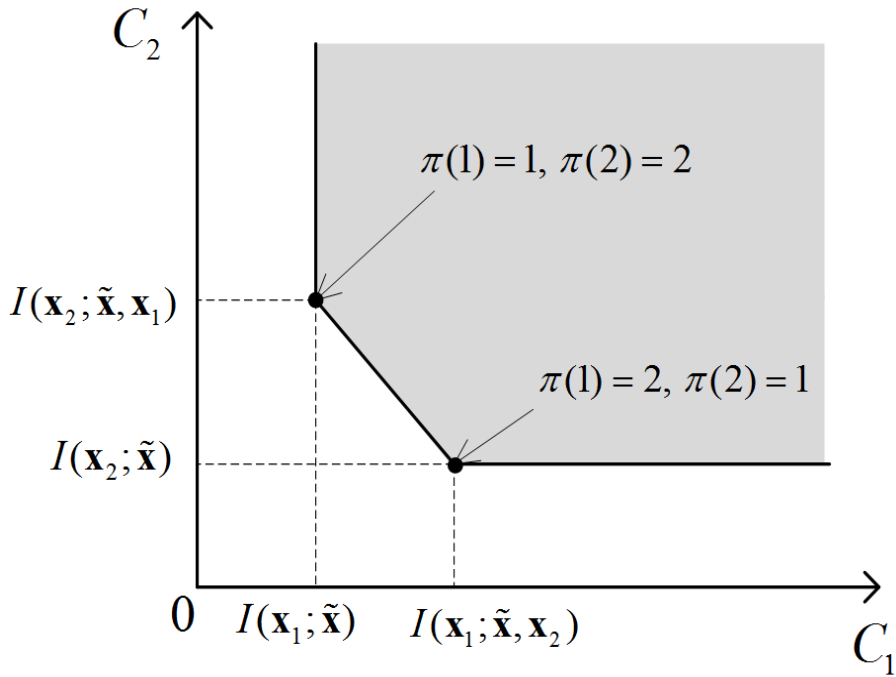}

\caption{\label{fig:corner points}Illustrative example of the contrapolymatroidal
region of backhaul capacities $(C_{1},C_{2})$ satisfying the constraint
(\ref{eq:original problem backhaul}) imposed by multivariate compression
for $N_{B}=2$. The two corner points are given by (\ref{eq:corner points}).}
\end{figure}

Lemma \ref{lem:contra polymatroid} shows that the region of backhaul
capacities that guarantees correct delivery of the compressed signals
(\ref{eq:Gaussian test channel each BS}) to the BSs, as identified
in Lemma \ref{lem:multivariate}, is a type of polyhedron known as
contrapolymatroid, as exemplified in Fig. \ref{fig:corner points}.
A specific feature of contrapolymatroid is that the corner points
can be easily characterized as in (\ref{eq:corner points}). From
standard rate-distortion theory arguments, the equality between (\ref{eq:compression rate MMSE compression})
and (\ref{eq:corner points}) implies that the corner point $(C_{\pi(1)},\ldots,C_{\pi(N_{B})})$
can be obtained for any permutation $\pi$ by the successive estimation-compression
strategy outlined above and illustrated in Fig. \ref{fig:successive comp}.

Overall, the discussion above shows that, for any correlation matrix
$\mathbf{\Omega}$ in (\ref{eq:compression covariance}), multivariate
compression is feasible by using the successive estimation-compression
architecture in Fig. \ref{fig:successive comp} if the backhaul capacities
satisfy the corner point condition (\ref{eq:corner points}) for the
given BS order $\pi$. Note that, in general, conditions (\ref{eq:corner points})
are more restrictive than (\ref{eq:original problem backhaul}), which
allows for any backhaul capacities in the contrapolymatroid. This
is because the solution to the optimization problem (\ref{eq:original problem})
can be seen by contradiction to lie necessarily on the boundary of
the region (\ref{eq:original problem backhaul}) but possibly not
on the corner points. Further discussion on this point can be found
in Sec. \ref{sub:Practical-Implementation}, where we observe that,
in practice, this limitation is not critical.

\section{Joint Design of Precoding and Compression\label{sec:Joint}}

In this section, we aim at jointly optimizing the precoding matrix
$\mathbf{A}$ and the compression covariance $\mathbf{\Omega}$ by
solving problem (\ref{eq:original problem}). In Sec. \ref{sec:Separate},
we will then consider the generally suboptimal strategy in which the
precoding matrix is fixed according to standard techniques, such as
ZF \cite{RZhang}, MMSE \cite{Luo:JSAC} or weighted sum-rate maximizing
precoding \cite{Huang}\cite{Luo} by neglecting the compression noise,
and only the compression noise matrix $\mathbf{\Omega}$ is optimized.

\subsection{MM Algorithm\label{sub:MM-Algorithm}}

As mentioned, the optimization (\ref{eq:original problem}) is a non-convex
problem. To tackle this issue, we will now show how to obtain an efficient
scheme that is able to achieve a stationary point of problem (\ref{eq:original problem}).
To this end, we first make a change of variable by defining the variables
$\mathbf{R}_{k}\triangleq\mathbf{A}_{k}\mathbf{A}_{k}^{\dagger}$
for $k\in\mathcal{N_{M}}$. Then, we define the functions $f_{k}\left(\{\mathbf{R}_{j}\}_{j=1}^{N_{M}},\mathbf{\Omega}\right)$
and $g_{\mathcal{S}}\left(\{\mathbf{R}_{k}\}_{k=1}^{N_{M}},\mathbf{\Omega}\right)$
with respect to the variables $\{\mathbf{R}_{k}\}_{k=1}^{N_{M}}$
which are obtained by substituting $\mathbf{R}_{k}=\mathbf{A}_{k}\mathbf{A}_{k}^{\dagger}$
into the functions $f_{k}\left(\mathbf{A},\mathbf{\Omega}\right)$
and $g_{\mathcal{S}}\left(\mathbf{A},\mathbf{\Omega}\right)$ in problem
(\ref{eq:original problem}), respectively, and the transmit power
constraint as $\mathrm{tr}\left(\sum_{k=1}^{N_{M}}\mathbf{E}_{i}^{\dagger}\mathbf{R}_{k}\mathbf{E}_{i}+\mathbf{\Omega}_{i,i}\right)\leq P_{i}$
for $i\in\mathcal{N_{B}}$. The so-obtained problem over the variables
$\{\mathbf{R}_{k}\}_{k=1}^{N_{M}}$ and $\mathbf{\Omega}$ is still
non-convex due to the second term in $f_{k}\left(\{\mathbf{R}_{j}\}_{j=1}^{N_{M}},\mathbf{\Omega}\right)$
and the first term in $g_{\mathcal{S}}\left(\{\mathbf{R}_{k}\}_{k=1}^{N_{M}},\mathbf{\Omega}\right)$,
which are concave in the variables $\{\mathbf{R}_{k}\}_{k=1}^{N_{M}}$
and $\mathbf{\Omega}$. However, we observe that it falls into the
class of difference-of-convex (DC) problems \cite{Beck}. Among various
algorithms having desirable properties for the solution of DC problems
\cite{Beck}, we adopt the Majorization Minimization (MM) scheme \cite{Beck},
which solves a sequence of convex problems obtained by linearizing
non-convex parts in the objective function $f_{k}\left(\{\mathbf{R}_{j}\}_{j=1}^{N_{M}},\mathbf{\Omega}\right)$
and the constraint function $g_{\mathcal{S}}\left(\{\mathbf{R}_{k}\}_{k=1}^{N_{M}},\mathbf{\Omega}\right)$.
It is shown that the MM algorithm converges to a stationary point
of the original non-convex problems (see, e.g., \cite[Theorem 1]{Luo},
\cite[Sec. 1.3.3]{Beck} and \cite[Theorem 3]{Scutari}). The proposed
algorithm is summarized in Table Algorithm 1 where we define the functions
$f_{k}^{\prime}(\{\mathbf{R}_{j}^{(t+1)}\}_{j=1}^{N_{M}},\mathbf{\Omega}^{(t+1)},\{\mathbf{R}_{j}^{(t)}\}_{j=1}^{N_{M}},\mathbf{\Omega}^{(t)})$
and $g_{\mathcal{S}}^{\prime}(\{\mathbf{R}_{j}^{(t+1)}\}_{j=1}^{N_{M}},\mathbf{\Omega}^{(t+1)},\{\mathbf{R}_{j}^{(t)}\}_{j=1}^{N_{M}},\mathbf{\Omega}^{(t)})$
as
\begin{align}
 & f_{k}^{\prime}\left(\{\mathbf{R}_{j}^{(t+1)}\}_{j=1}^{N_{M}},\mathbf{\Omega}^{(t+1)},\{\mathbf{R}_{j}^{(t)}\}_{j=1}^{N_{M}},\mathbf{\Omega}^{(t)}\right)\\
\triangleq & \log\det\left(\mathbf{I}+\mathbf{H}_{k}\left(\sum_{j=1}^{N_{M}}\mathbf{R}_{j}^{(t+1)}+\mathbf{\Omega}^{(t+1)}\right)\mathbf{H}_{k}^{\dagger}\right)\nonumber \\
- & \varphi\left(\mathbf{I}+\mathbf{H}_{k}\left(\sum_{j=1,j\neq k}^{N_{M}}\mathbf{R}_{j}^{(t+1)}+\mathbf{\Omega}^{(t+1)}\right)\mathbf{H}_{k}^{\dagger},\,\,\mathbf{I}+\mathbf{H}_{k}\left(\sum_{j=1,j\neq k}^{N_{M}}\mathbf{R}_{j}^{(t)}+\mathbf{\Omega}^{(t)}\right)\mathbf{H}_{k}^{\dagger}\right)\nonumber
\end{align}
and
\begin{align}
 & g_{\mathcal{S}}^{\prime}\left(\{\mathbf{R}_{j}^{(t+1)}\}_{j=1}^{N_{M}},\mathbf{\Omega}^{(t+1)},\{\mathbf{R}_{j}^{(t)}\}_{j=1}^{N_{M}},\mathbf{\Omega}^{(t)}\right)\\
\triangleq & \varphi\left(\sum_{j=1}^{N_{M}}\mathbf{E}_{i}^{\dagger}\mathbf{R}_{j}^{(t+1)}\mathbf{E}_{i}+\mathbf{\Omega}_{i,i}^{(t+1)},\,\,\sum_{j=1}^{N_{M}}\mathbf{E}_{i}^{\dagger}\mathbf{R}_{j}^{(t)}\mathbf{E}_{i}+\mathbf{\Omega}_{i,i}^{(t)}\right)\nonumber \\
- & \log\det\left(\mathbf{E}_{\mathcal{S}}^{\dagger}\mathbf{\Omega}^{(t+1)}\mathbf{E}_{\mathcal{S}}\right)\nonumber
\end{align}
with the function $\varphi(\mathbf{X},\mathbf{Y})$ given as
\begin{equation}
\varphi(\mathbf{X},\mathbf{Y})\triangleq\log\det\left(\mathbf{Y}\right)+\frac{1}{\ln2}\mathrm{tr}\left(\mathbf{Y}^{-1}\left(\mathbf{X}-\mathbf{Y}\right)\right).
\end{equation}

\begin{algorithm}
\caption{MM Algorithm for problem (\ref{eq:original problem})}

1. Initialize the matrices $\{\mathbf{R}_{k}^{(1)}\}_{k=1}^{N_{M}}$
and $\mathbf{\Omega}^{(1)}$ to arbitrary feasible positive semidefinite
matrices for problem (\ref{eq:original problem}) and set $t=1$.

2. Update the matrices $\{\mathbf{R}_{k}^{(t+1)}\}_{k=1}^{N_{M}}$
and $\mathbf{\Omega}^{(t+1)}$ as a solution of the following (convex)
problem.
\begin{align}
\underset{\{\mathbf{R}_{k}^{(t+1)}\succeq\mathbf{0}\}_{k=1}^{N_{M}},\mathbf{\Omega}^{(t+1)}\succeq\mathbf{0}}{\mathrm{maximize}} & \sum_{k=1}^{N_{M}}w_{k}f_{k}^{\prime}\left(\{\mathbf{R}_{j}^{(t+1)}\}_{j=1}^{N_{M}},\mathbf{\Omega}^{(t+1)},\{\mathbf{R}_{j}^{(t)}\}_{j=1}^{N_{M}},\mathbf{\Omega}^{(t)}\right)\label{eq:convex approximated problem JDD}\\
\mathrm{s.t.}\,\,\,\,\,\,\,\,\, & g_{\mathcal{S}}^{\prime}\left(\{\mathbf{R}_{j}^{(t+1)}\}_{j=1}^{N_{M}},\mathbf{\Omega}^{(t+1)},\{\mathbf{R}_{j}^{(t)}\}_{j=1}^{N_{M}},\mathbf{\Omega}^{(t)}\right)\leq\sum_{i\in\mathcal{S}}C_{i},\,\,\mathrm{for\, all}\,\mathcal{S}\subseteq\mathcal{N_{B}},\nonumber \\
 & \mathrm{tr}\left(\sum_{k=1}^{N_{M}}\mathbf{E}_{i}^{\dagger}\mathbf{R}_{k}^{(t+1)}\mathbf{E}_{i}+\mathbf{\Omega}_{i,i}^{(t+1)}\right)\leq P_{i},\,\,\mathrm{for\, all}\, i\in\mathcal{N_{B}}.\nonumber
\end{align}

3. Go to Step 4 if a convergence criterion is satisfied. Otherwise,
set $t\leftarrow t+1$ and go back to Step 2.

4. Calculate the precoding matrices $\mathbf{A}_{k}\leftarrow\mathbf{V}_{k}\mathbf{D}_{k}^{1/2}$
for $k\in\mathcal{N_{M}}$, where $\mathbf{D}_{k}$ is a diagonal
matrix whose diagonal elements are the nonzero eigenvalues of $\mathbf{R}_{k}^{(t)}$
and the columns of $\mathbf{V}_{k}$ are the corresponding eigenvectors.
\end{algorithm}

\subsection{Practical Implementation\label{sub:Practical-Implementation}}

As we have discussed in Sec. \ref{sec:Problem formulation}, given
the solution $(\mathbf{A},\mathbf{\Omega})$ obtained from the proposed
algorithm, the central processor should generally perform joint compression
in order to obtain the signals $\mathbf{x}_{i}$ to be transmitted
by the BSs. However, as seen in Sec. \ref{sub:Successive-Compression},
if the solution is such that the corner point conditions (\ref{eq:corner points})
are satisfied for a given permutation $\pi$ of the BSs' indices,
then the simpler successive estimation-compression structure of Fig.
\ref{fig:successive comp} can be leveraged instead. We recall, from
Lemma \ref{lem:contra polymatroid}, that in order to check whether
the conditions (\ref{eq:corner points}) are satisfied for some order
$\pi$, it is sufficient to observe which inequalities (\ref{eq:original problem backhaul})
are satisfied with equality: If these inequalities correspond to the
subsets $\mathcal{S}=\{\pi(1)\},\{\pi(1),\pi(2)\},\ldots,\{\pi(1),\ldots,\pi(N_{B})\}$
for a given permutation $\pi$, then the given solution corresponds
to the corner point (\ref{eq:corner points}) with the given $\pi$.
In our extensive numerical results, we have consistently found this
condition to be verified. As a result, in practice, one can implement
the compression strategy characterized by the calculated covariance
$\mathbf{\Omega}$ by employing the implementation of Fig. \ref{fig:successive comp}
with the obtained ordering $\pi$.

\subsection{Independent Quantization\label{sub:Independent-Quantization}}

For reference, it is useful to consider the weighted sum-rate maximization
problem with independent quantization noises as in \cite{Simeone}-\cite{Hong}.
This is formulated as (\ref{eq:original problem}) with the additional
constraints
\begin{equation}
\mathbf{\Omega}_{i,j}=\mathbf{0},\,\,\mathrm{for\,\, all}\,\, i\neq j\in\mathcal{N_{B}}.\label{eq:independent quantization}
\end{equation}
Since the constraints (\ref{eq:independent quantization}) are affine,
the MM algorithm in Table Algorithm 1 is still applicable by simply
setting to zero matrices $\mathbf{\Omega}_{i,j}=\mathbf{0}$ for $i\neq j$
as per (\ref{eq:independent quantization}).

\subsection{Robust Design with Imperfect CSI\label{sub:Robust-Design-with}}

So far, we have assumed that the central encoder has information about
the global channel matrices $\mathbf{H}_{k}$ for $k\in\mathcal{N_{M}}$.
In this subsection, we discuss the robust design of the precoding
matrix $\mathbf{A}$ and the compression covariance $\mathbf{\Omega}$
in the presence of uncertainty at the central encoder regarding the
channel matrices $\mathbf{H}_{k}$ for $k\in\mathcal{N_{M}}$. Specifically,
we focus on deterministic worst-case optimization under two different
uncertainty models, namely singular value uncertainty \cite{Loyka}
and ellipsoidal uncertainty models (see \cite{Shen}\cite[Sec. 4.1]{Bjornson}
and references therein). While the singular value uncertainty model
can be related via appropriate bounds to any normed uncertainty on
the channel matrices, as discussed in \cite[Sec. V]{Loyka}, the ellipsoidal
uncertainty model is more accurate when knowledge of the covariance
matrix of the CSI error, due, e.g., to estimation, is available \cite[Sec. 4.1]{Bjornson}.
In the following, we briefly discuss both models.

\subsubsection{Singular Value Uncertainty Model}

Considering multiplicative uncertainty model of \cite[Sec. II-A]{Loyka},
the actual channel matrix $\mathbf{H}_{k}$ toward each MS $k$ is
modeled as
\begin{equation}
\mathbf{H}_{k}=\hat{\mathbf{H}}_{k}\left(\mathbf{I}+\mathbf{\mathbf{\Delta}}_{k}\right),\label{eq:multiplicative model}
\end{equation}
where the matrix $\hat{\mathbf{H}}_{k}$ is the CSI known at the central
encoder and the matrix $\mathbf{\mathbf{\Delta}}_{k}\in\mathbb{C}^{n_{B}\times n_{B}}$
accounts for the multiplicative uncertainty matrix. The latter is
bounded as
\begin{equation}
\sigma_{\mathrm{max}}\left(\mathbf{\mathbf{\Delta}}_{k}\right)\leq\epsilon_{k}<1,\label{eq:singular value uncertainty}
\end{equation}
where $\sigma_{\mathrm{max}}(\mathbf{X})$ is the largest singular
value of matrix $\mathbf{X}$. Then, the problem of interest is to
maximizing the worst-case weighted sum-rate over all possible uncertainty
matrices $\mathbf{\mathbf{\Delta}}_{k}$ for $k\in\mathcal{N_{M}}$
subject to the backhaul capacity (\ref{eq:original problem backhaul})
and power constraints (\ref{eq:original problem power}), namely\begin{subequations}\label{eq:worst-case problem}
\begin{align}
\underset{\mathbf{A},\,\mathbf{\Omega}\succeq\mathbf{0}}{\mathrm{maximize}} & \,\,\min_{\left\{ \mathbf{\mathbf{\Delta}}_{k}\,\,\mathrm{s.t.}\,\,(\ref{eq:singular value uncertainty})\right\} {}_{k=1}^{N_{M}}}\sum_{k=1}^{N_{M}}w_{k}f_{k}\left(\mathbf{A},\mathbf{\Omega}\right)\label{eq:worst-case problem objective}\\
\mathrm{s.t.}\,\,\,\,\,\,\,\,\, & g_{\mathcal{S}}\left(\mathbf{A},\mathbf{\Omega}\right)\leq\sum_{i\in\mathcal{S}}C_{i},\,\,\mathrm{for\,\, all}\,\,\mathcal{S}\subseteq\mathcal{N_{B}},\label{eq:worst-case problem backhaul}\\
 & \mathrm{tr}\left(\mathbf{E}_{i}^{\dagger}\mathbf{A}\mathbf{A}^{\dagger}\mathbf{E}_{i}+\mathbf{\Omega}_{i,i}\right)\leq P_{i},\,\,\mathrm{for\,\, all}\,\, i\in\mathcal{N_{B}}.\label{eq:worst-case problem power}
\end{align}
\end{subequations} The following lemma offers an equivalent formulation
for problem (\ref{eq:worst-case problem}).

\begin{lemma}\label{lem:problem singular uncertainty}

The problem (\ref{eq:worst-case problem}) is equivalent to the problem
(\ref{eq:original problem}) with the channel matrix $\mathbf{H}_{k}$
replaced with $(1-\epsilon_{k})\hat{\mathbf{H}}_{k}$ for $k\in\mathcal{N_{M}}$.
\end{lemma}

\begin{proof} We first observe that the uncertainty matrix $\mathbf{\mathbf{\Delta}}_{k}$
affects only the corresponding rate function $f_{k}(\mathbf{A},\mathbf{\Omega})$
in (\ref{eq:rate MS k computed}). Therefore, the minimization versus
matrices $\mathbf{\mathbf{\Delta}}_{k}$ for $k\in\mathcal{N_{M}}$
in (\ref{eq:worst-case problem objective}) can be performed separately
for each $k$ by solving the problem $\min_{\mathbf{\mathbf{\Delta}}_{k}}f_{k}(\mathbf{A},\mathbf{\Omega})$.
It can be now easily seen, following \cite[Theorem 1]{Loyka}, that
the result of this minimization is obtained when $\mathbf{\mathbf{\Delta}}_{k}$
is such that $\mathbf{\mathbf{\Delta}}_{k}=-\epsilon_{k}\mathbf{I}$.
This concludes the proof.\end{proof}

Based on Lemma \ref{lem:problem singular uncertainty}, one can hence
solve problem (\ref{eq:worst-case problem}) by using the MM algorithm
in Table Algorithm 1 with only change of the channel matrices from
$\{\mathbf{H}_{k}\}_{k=1}^{N_{M}}$ to $\{(1-\epsilon_{k})\hat{\mathbf{H}}_{k}\}_{k=1}^{N_{M}}$.

\subsubsection{Ellipsoidal Uncertainty Model}

We now consider the ellipsoidal uncertainty model. To this end, for
simplicity and following related literature (see, e.g., \cite[Sec. 4.1]{Bjornson}),
we focus on multiple-input single-output (MISO) case where each MS
is equipped with a single antenna, i.e., $n_{M,k}=1$ for $k\in\mathcal{N_{M}}$.
Thus, we denote the channel vector corresponding to each MS $k$ by
$\mathbf{H}_{k}=\mathbf{h}_{k}^{\dagger}\in\mathbb{C}^{1\times n_{B}}$.
The actual channel $\mathbf{h}_{k}$ is then modeled as
\begin{equation}
\mathbf{h}_{k}=\hat{\mathbf{h}}_{k}+\mathbf{e}_{k},\label{eq:channel error}
\end{equation}
with $\hat{\mathbf{h}}_{k}$ and $\mathbf{e}_{k}$ being the presumed
CSI available at the central encoder and the CSI error, respectively.
The error vector $\mathbf{e}_{k}$ is assumed to be bounded within
the ellipsoidal region described as
\begin{equation}
\mathbf{e}_{k}^{\dagger}\mathbf{C}_{k}\mathbf{e}_{k}\leq1,\label{eq:elliptic model}
\end{equation}
for $k\in\mathcal{N_{M}}$ with the matrix $\mathbf{C}_{k}\succ\mathbf{0}$
specifying the size and shape of the ellipsoid \cite{Shen}.

Following the standard formulation, we consider here the ``dual''
problem of power minimization under signal-to-interference-plus-noise
ratio (SINR) constraints for all MSs (see \cite[Sec. 4.1]{Bjornson}
and references therein). This problem is stated as\begin{subequations}\label{eq:problem Pmin}
\begin{align}
\underset{\{\mathbf{R}_{k}\succeq\mathbf{0}\}_{k=1}^{N_{M}},\mathbf{\Omega}\succeq\mathbf{0}}{\mathrm{minimize}} & \sum_{i=1}^{N_{B}}\mu_{i}\cdot\mathrm{tr}\left(\sum_{k=1}^{N_{M}}\mathbf{E}_{i}^{\dagger}\mathbf{R}_{k}\mathbf{E}_{i}+\mathbf{\Omega}_{i,i}\right)\label{eq:problem Pmin objective}\\
\mathrm{s.t.}\,\,\,\, & \frac{\mathbf{h}_{k}^{\dagger}\mathbf{R}_{k}\mathbf{h}_{k}}{\sum_{j\in\mathcal{N_{M}}\setminus\{k\}}\mathbf{h}_{k}^{\dagger}\mathbf{R}_{j}\mathbf{h}_{k}+\mathbf{h}_{k}^{\dagger}\mathbf{\Omega}\mathbf{h}_{k}+1}\geq\Gamma_{k},\,\,\mathrm{for}\,\,\mathrm{all}\,\,\mathbf{e}_{k}\,\,\mathrm{with}\,\,(\ref{eq:elliptic model})\,\,\mathrm{and}\,\, k\in\mathcal{N_{M}},\label{eq:problem Pmin SINR constraint}\\
 & g_{\mathcal{S}}\left(\mathbf{A},\mathbf{\Omega}\right)\leq\sum_{i\in\mathcal{S}}C_{i},\,\,\mathrm{for\,\, all}\,\,\mathcal{S}\subseteq\mathcal{N_{B}},\label{eq:problem Pmin backhaul constraints}
\end{align}
\end{subequations}where the coefficients $\mu_{i}\geq0$ are arbitrary
weights, $\Gamma_{k}$ is the SINR constraint for MS $k$, and we
recall that we have $\mathbf{R}_{k}\triangleq\mathbf{A}_{k}\mathbf{A}_{k}^{\dagger}$
for $k\in\mathcal{N_{M}}$. The problem (\ref{eq:problem Pmin}) is
challenging since it contains infinite number of constraints in (\ref{eq:problem Pmin SINR constraint}).
But, following the conventional \textit{S-procedure} \cite[Appendix B.2]{Boyd},
we can translate the constraints (\ref{eq:problem Pmin SINR constraint})
into a finite number of linear constraints by introducing auxiliary
variables $\beta_{k}$ for $k\in\mathcal{N_{M}}$, as discussed in
the following lemma.

\begin{lemma}\label{lem:S lemma} The constraints (\ref{eq:problem Pmin SINR constraint})
hold if and only if there exist constants $\{\beta_{k}\geq0\}_{k=1}^{N_{M}}$
such that the condition
\begin{equation}
\left[\begin{array}{cc}
\mathbf{\Xi}_{k} & \mathbf{\Xi}_{k}\hat{\mathbf{h}}_{k}\\
\hat{\mathbf{h}}_{k}^{\dagger}\mathbf{\Xi}_{k} & \hat{\mathbf{h}}_{k}^{\dagger}\mathbf{\Xi}_{k}\hat{\mathbf{h}}_{k}-\Gamma_{k}
\end{array}\right]-\beta_{k}\left[\begin{array}{cc}
\mathbf{C}_{k} & \mathbf{0}\\
\mathbf{0} & -1
\end{array}\right]\succeq\mathbf{0}\label{eq:S Lemma constraint}
\end{equation}
is satisfied for all $k\in\mathcal{N_{M}}$ where we have defined
$\mathbf{\Xi}_{k}=\mathbf{R}_{k}-\Gamma_{k}\sum_{j\in\mathcal{N_{M}}\setminus\{k\}}\mathbf{R}_{j}-\Gamma_{k}\mathbf{\Omega}$
for $k\in\mathcal{N_{M}}$. \end{lemma}

\begin{proof}It directly follows by applying the S-procedure \cite[Appendix B.2]{Boyd}.
\end{proof}

By transforming the constraint (\ref{eq:problem Pmin SINR constraint})
into (\ref{eq:S Lemma constraint}), we obtain a problem that falls
again in the class of DC problems \cite{Beck}. Therefore, one can
easily derive the MM algorithm, similar to Table Algorithm 1, by linearizing
the non-convex terms in the constraint (\ref{eq:problem Pmin backhaul constraints}).
The algorithm is guaranteed to converge to a stationary point of problem
(\ref{eq:problem Pmin}) (see, e.g., \cite[Theorem 1]{Luo}, \cite[Sec. 1.3.3]{Beck}
and \cite[Theorem 3]{Scutari}).

\section{Separate Design of Precoding and Compression\label{sec:Separate}}

In this section, we discuss a simpler approach in which the precoding
matrix $\mathbf{A}$ is fixed a priori to some standard scheme, such
as ZF, MMSE or weighted sum-rate maximizing precoding, by neglecting
the compression noise. The compression covariance $\mathbf{\Omega}$
is then designed separately so as to maximize the weighted sum-rate.

\subsection{Selection of the Precoding Matrix\label{sub:separate Precoding}}

The precoding matrix $\mathbf{A}$ is first selected according to
some standard criterion \cite{RZhang}-\cite{Luo} by neglecting the
compression noise. A subtle issue arises when selecting the precoding
matrix $\mathbf{A}$ that requires some discussion. Specifically,
the design of $\mathbf{A}$ should be done by assuming a reduced power
constraint, say $\gamma_{i}P_{i}$ for some $\gamma_{i}\in(0,1)$
for $i\in\mathcal{N_{B}}$. The power offset factor $\gamma_{i}\in(0,1)$
is necessary since the final signal $\mathbf{x}_{i}$ transmitted
by each BS $i$ is given by (\ref{eq:Gaussian test channel each BS})
and is thus the sum of the precoded signal $\mathbf{E}_{i}^{\dagger}\mathbf{A}\mathbf{s}$
and the compression noise $\mathbf{q}_{i}$. Therefore, if the power
of the precoded part $\mathbf{E}_{i}^{\dagger}\mathbf{A}\mathbf{s}$
is selected to be equal to the power constraint $P_{i}$, the compression
noise power would be forced to be zero. But this is possible only
when the backhaul capacity grows to infinity due to (\ref{eq:original problem backhaul}).
As a result, in order to make the compression feasible, one needs
to properly select the parameters $\gamma_{1},\ldots,\gamma_{N_{B}}$
depending on the backhaul constraints.

\subsection{Optimization of the Compression Covariance\label{sub:separate Compression cov}}

Having fixed the precoding matrix $\mathbf{A}$, the problem then
reduces to solving problem (\ref{eq:original problem}) only with
respect to the compression covariance $\mathbf{\Omega}$. The obtained
problem is thus a DC problem which can be tackled via the MM algorithm
described in Table Algorithm 1 by limiting the optimization at Step
2 only to matrix $\mathbf{\Omega}$. It is observed that, as discussed
above, this problem may not be feasible if the parameters $\gamma_{i}$,
$i\in\mathcal{N_{B}}$, are too large. In practice, one can set these
parameters using various search strategies such as bisection.

\section{Numerical Results\label{sec:Numerical-Results}}

In this section, we present numerical results in order to investigate
the advantage of the proposed approach based on multivariate compression
and on the joint design of precoding and compression as compared to
the conventional approaches based on independent compression across
the BSs and separate design. We will focus on the sum-rate performance
$R_{\mathrm{sum}}=\sum_{k\in\mathcal{N_{M}}}R_{k}$ (i.e., we set
$w_{k}=1$ in (\ref{eq:original problem objective})). We also assume
that there is one MS active in each cell and we consider three cells,
so that we have $N_{B}=N_{M}=3$. Every BS is subject to the same
power constraint $P$ and has the same backhaul capacity $C$, i.e.,
$P_{i}=P$ and $C_{i}=C$ for $i\in\mathcal{N_{B}}$.

\subsection{Wyner Model}

\begin{figure}
\centering\includegraphics[width=12cm,height=9cm]{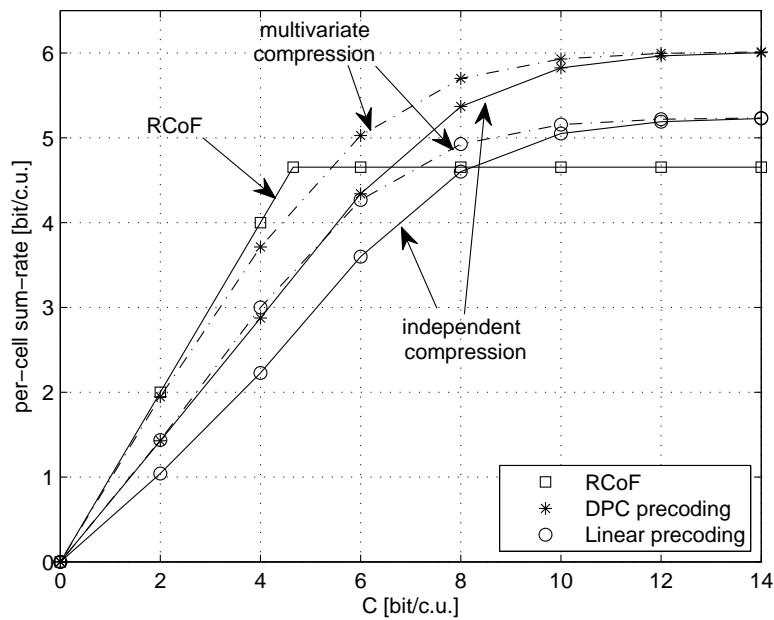}

\caption{\label{fig:graph Wyner}Per-cell sum-rate versus the backhaul capacity
$C$ for the circular Wyner model \cite{Gesbert} with $P=20$ dB
and $g=0.5$.}
\end{figure}

We start by considering as a benchmark the performance in a simple
circulant Wyner model. In this model, all MSs and BSs have a single
antenna and the channel matrices $\mathbf{H}_{k,j}$ reduce to deterministic
scalars given as $\mathbf{H}_{k,k}=1$ for $k=1,2,3$ and $\mathbf{H}_{k,j}=g\in[0,1]$
for $j\neq k$ \cite{Gesbert}. In Fig. \ref{fig:graph Wyner}, we
compare the proposed scheme with joint design of precoding and compression
with state-of-the-art techniques, namely the compressed DPC of \cite{Simeone},
which corresponds to using DPC precoding with independent quantization,
and reverse Compute-and-Forward (RCoF) \cite{Hong}. We also show
the performance with linear precoding for reference. It is observed
that multivariate compression significantly outperforms the conventional
independent compression strategy for both linear and DPC precoding.
Moreover, RCoF in \cite{Hong} remains the most effective approach
in the regime of moderate backhaul $C$, although multivariate compression
allows to compensate for most of the rate loss of standard DPC precoding
in the low-backhaul regime%
\footnote{The saturation of the rate of RCoF for sufficiently large $C$ is
due to the integer constraints imposed on the function of the messages
to be computed by the MSs \cite{Hong}.%
}.

\subsection{General Fading Model}

In this subsection, we evaluate the average sum-rate performance as
obtained by averaging the sum-rate $R_{\mathrm{sum}}$ over the the
realization of the fading channel matrices. The elements of the channel
matrix $\mathbf{H}_{k,i}$ between the MS in the $k$th cell and the
BS in the $i$th cell are assumed to be i.i.d. complex Gaussian random
variables with $\mathcal{CN}(0,\alpha^{|i-k|})$ in which we call
$\alpha$ the inter-cell channel gain. Moreover, each BS is assumed
to use two transmit antennas while each MS is equipped with a single
receive antenna. In the separate design, we assume that the precoding
matrix $\mathbf{A}$ is obtained via the sum-rate maximization scheme
in \cite{Huang} under the power constraint $\gamma P$ for each BS
with $\gamma\in(0,1)$ selected as discussed in Sec. \ref{sub:separate Precoding}.
Note that the algorithm of \cite{Huang} finds a stationary point
for the sum-rate maximization problem using the MM approach, similar
to Table Algorithm 1 without consideration of the quantization noises.

\begin{figure}
\centering\includegraphics[width=12cm,height=9cm]{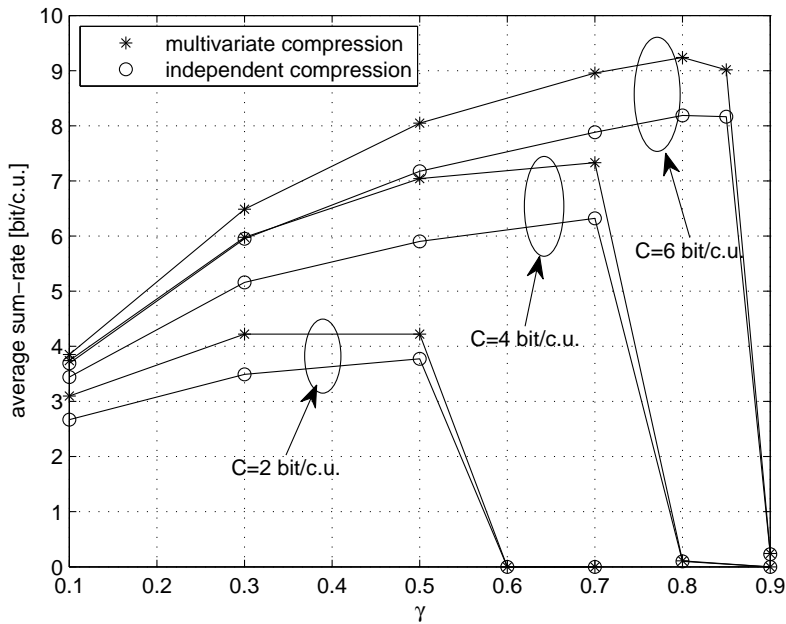}

\caption{\label{fig:graph Gamma}Average sum-rate versus the power offset factor
$\gamma$ for the separate design of linear precoding and compression
in Sec. \ref{sec:Separate} with $P=5$ dB and $\alpha=0$ dB.}
\end{figure}

Fig. \ref{fig:graph Gamma} demonstrates the impact of the power offset
factor $\gamma$ on the separate design of linear precoding and compression
described in Sec. \ref{sec:Separate} with $P=5$ dB and $\alpha=0$
dB. Increasing $\gamma$ means that more power is available at each
BS, which generally results in a better sum-rate performance. However,
if $\gamma$ exceeds some threshold value, the sum-rate is significantly
degraded since the problem of optimizing the compression covariance
$\mathbf{\Omega}$ given the precoder $\mathbf{A}$ is more likely
to be infeasible as argued in Sec. \ref{sub:separate Precoding}.
This threshold value grows with the backhaul capacity, since a larger
backhaul capacity allows for a smaller power of the quantization noises.
Throughout the rest of this section, the power offset factor $\gamma$
is optimized via numerical search.

\begin{figure}
\centering\includegraphics[width=12cm,height=9cm]{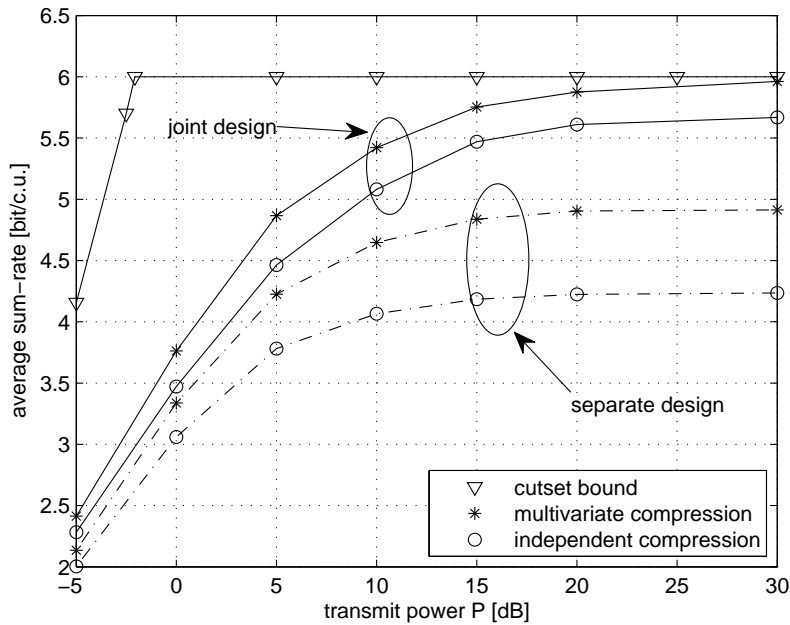}

\caption{\label{fig:graph SNR}Average sum-rate versus the transmit power $P$
for linear precoding with $C=2$ bit/c.u. and $\alpha=0$ dB.}
\end{figure}

In Fig. \ref{fig:graph SNR}, the average sum-rate performance of
the linear precoding and compression schemes is plotted versus the
transmit power $P$ with $C=2$ bit/c.u. and $\alpha=0$ dB. It is
seen that the gain of multivariate compression is more pronounced
when each BS uses a larger transmit power. This implies that, as the
received SNR increases, more efficient compression strategies are
called for. In a similar vein, the importance of the joint design
of precoding and compression is more significant when the transmit
power is larger. Moreover, it is seen that multivariate compression
is effective in partly compensating for the suboptimality of the separate
design. For reference, we also plot the cutset bound which is obtained
as $\min\{R_{\mathrm{full}},3C\}$ where $R_{\mathrm{full}}$ is the
sum-capacity achievable when the BSs can fully cooperate under per-BS
power constraints. We have obtained the rate $R_{\mathrm{full}}$
by using the inner-outer iteration algorithm proposed in \cite[Sec. II]{Huh}.
It is worth noting that only the proposed joint design with multivariate
compression approaches the cutset bound as the transmit power increases.

\begin{figure}
\centering\includegraphics[width=12cm,height=9cm]{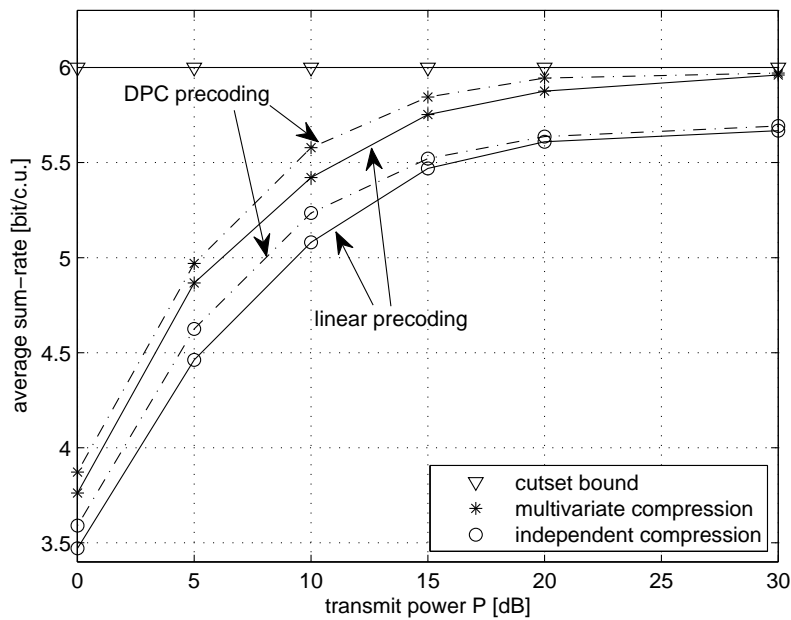}

\caption{\label{fig:graph DPC}Average sum-rate versus the transmit power $P$
for the joint design in Sec. \ref{sec:Joint} with $C=2$ bit/c.u.
and $\alpha=0$ dB.}
\end{figure}

In Fig. \ref{fig:graph DPC}, we compare two precoding methods, DPC
and linear precoding, by plotting the average sum-rate versus the
transmit power $P$ for the joint design in Sec. \ref{sec:Joint}
with $C=2$ bit/c.u. and $\alpha=0$ dB. For DPC, we have applied
Algorithm 1 with a proper modification for all permutations $\pi_{\mathrm{DPC}}$
of MSs' indices $\mathcal{N_{M}}$ and took the largest sum-rate.
Unlike the conventional broadcast channels with perfect backhaul links
where there exists constant sum-rate gap between DPC and linear precoding
at high SNR (see, e.g., \cite{Lee}), Fig. \ref{fig:graph DPC} shows
that DPC is advantageous only in the regime of intermediate $P$ due
to the limited-capacity backhaul links. This implies that the overall
performance is determined by the compression strategy rather than
precoding method when the backhaul capacity is limited at high SNR.

\begin{figure}
\centering\includegraphics[width=12cm,height=9cm]{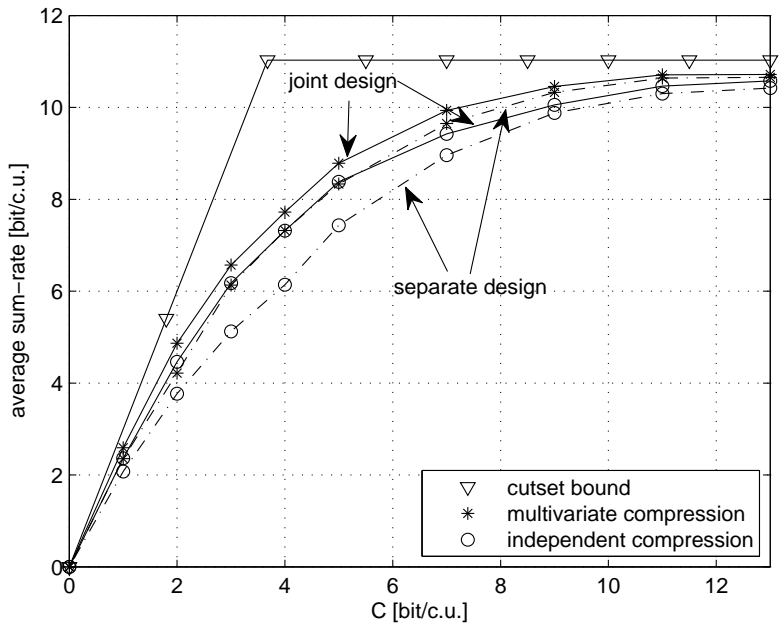}

\caption{\label{fig:graph Ci}Average sum-rate versus the backhaul capacity
$C$ for linear precoding with $P=5$ dB and $\alpha=0$ dB.}
\end{figure}

Fig. \ref{fig:graph Ci} plots the average sum-rate versus the backhaul
capacity $C$ for linear precoding with $P=5$ dB and $\alpha=0$
dB. It is observed that when the backhaul links have enough capacity,
the benefits of multivariate compression or joint design of precoding
and compression become negligible since the overall performance becomes
limited by the sum-capacity achievable when the BSs are able to fully
cooperate with each other. It is also notable that the separate design
with multivariate compression outperforms the joint design with independent
quantization for backhaul capacities larger than $5$ bit/c.u.

\begin{figure}
\centering\includegraphics[width=12cm,height=9cm]{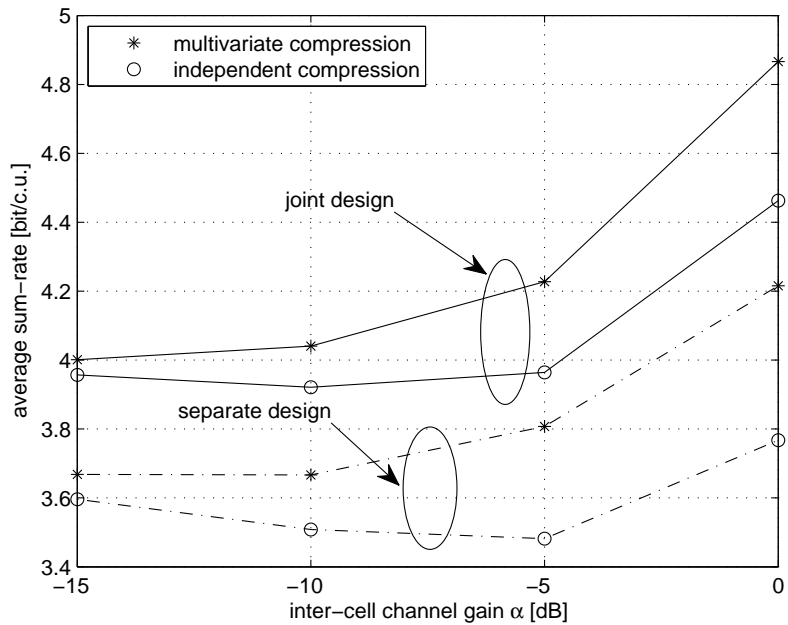}

\caption{\label{fig:graph alpha}Average sum-rate versus the inter-cell channel
gain $\alpha$ for linear precoding with $C=2$ bit/c.u. and $P=5$
dB.}
\end{figure}

Finally, we plot the sum-rate versus the inter-cell channel gain $\alpha$
for linear precoding with $C=2$ bit/c.u. and $P=5$ dB in Fig. \ref{fig:graph alpha}.
We note that the multi-cell system under consideration approaches
the system consisting of $N_{B}$ parallel single-cell networks as
the inter-cell channel gain $\alpha$ decreases. Thus, the advantage
of multivariate compression is not significant for small values of
$\alpha$, since introducing correlation of the quantization noises
across BSs is helpful only when each MS suffers from a superposition
of quantization noises emitted from multiple BSs.

\section{Conclusions\label{sec:Conclusions}}

In this work, we have studied the design of joint precoding and compression
strategies for the downlink of cloud radio access networks where the
BSs are connected to the central encoder via finite-capacity backhaul
links. Unlike the conventional approaches where the signals corresponding
to different BSs are compressed independently, we have proposed to
exploit multivariate compression of the signals of different BSs in
order to control the effect of the additive quantization noises at
the MSs. The problem of maximizing the weighted sum-rate subject to
power and backhaul constraints was formulated, and an iterative MM
algorithm was proposed that achieves a stationary point of the problem.
Moreover, we have proposed a novel way of implementing multivariate
compression that does not require joint compression of all the BSs'
signals but is based on successive per-BS estimation-compression steps.
Robust design with imperfect CSI was also discussed. Via numerical
results, it was confirmed that the proposed approach based on multivariate
compression and on joint precoding and compression strategy outperforms
the conventional approaches based on independent compression and separate
design of precoding and compression strategies. This is especially
true when the transmit power or the inter-cell channel gain are large,
and when the limitation imposed by the finite-capacity backhaul links
is significant.

\end{document}